%% file: bernpoly.tex
\documentclass[12pt]{article}
\usepackage{graphicx,verbatim,array,multicol,courier}
\usepackage{amsmath,amssymb,amsthm,mathrsfs}
\usepackage[pdftex,bookmarks,colorlinks]{hyperref}
\usepackage{fullpage}
\usepackage{topcapt}
\usepackage{chicago}
\usepackage{color}
\usepackage[section]{placeins} 
\usepackage{mathdots}
\usepackage{mathtools} 
\usepackage{epstopdf} 
\usepackage[labelformat=simple]{subfig}
\usepackage{soul} 
\usepackage{booktabs} 
\usepackage{dsfont}

\definecolor{navy}{rgb}{0,0,0.502}
\definecolor{brown}{rgb}{0.59, 0.29, 0.0}

\hypersetup{colorlinks,%
citecolor=blue,%
filecolor=blue,%
linkcolor=blue,%
urlcolor=blue,%
}

\input{characters.tex}

\numberwithin{equation}{section}

\theoremstyle{definition}
\newtheorem{defi}{Definition}[section]
\newtheorem{cond}[defi]{Condition}

\theoremstyle{plain}
\newtheorem{theo}[defi]{Theorem}
\newtheorem{prop}[defi]{Proposition}
\newtheorem{lem}[defi]{Lemma}

\theoremstyle{remark}
\newtheorem{rem}[defi]{Remark}

\newcommand{\brown}[1]{\bgroup\color{brown}{#1}\egroup}

\begin{document}

\title{Multivariate Nonparametric Estimation of the Pickands
Dependence Function using Bernstein Polynomials}
\author{
G. Marcon, S. A. Padoan, P. Naveau, P. Muliere and J. Segers
\footnote{Marcon is Post-doc at university of Pavia, Italy, 
E-mail: giulia.marcon@phd.unibocconi.it.
Muliere and Padoan work at the
Department of Decision Sciences,
Bocconi University of Milan, via Roentgen 1, 20136 Milano, Italy.
E-mail: pietro.muliere@unibocconi.it,
simone.padoan@unibocconi.it. Naveau is a CNRS researcher at the
Laboratoire des Sciences du Climat et l'Environnement, Gif-sur-Yvette, France.
E-mail: naveau@lsce.ipsl.fr. Segers is a Professor at the Universit\'{e}
catholique de Louvain, Institut de statistique, biostatistique et sciences actuarielles, Voie du
Roman Pays 20, B-1348 Louvain-la-Neuve, Belgium. E-mail: Johan.Segers@uclouvain.be.}
}

\maketitle

\begin{abstract}
Many applications in risk analysis, especially in environmental sciences,  require the estimation of the dependence among multivariate maxima.
A way to do this is by inferring the Pickands dependence function of the underlying extreme-value copula.
A nonparametric estimator is constructed as the sample equivalent of a multivariate extension of the madogram.
Shape constraints on the family of Pickands dependence functions are taken into account by means of a representation in terms of a specific type of Bernstein polynomials. 
The large-sample theory of the estimator is developed and its finite-sample performance is evaluated with a simulation study.
The approach is illustrated by analyzing  clusters consisting of  seven weather stations that have recorded weekly maxima of hourly rainfall in France from 1993 to 2011.
\\

\end{abstract}

\section{Introduction and background}\label{sec:intro}

In recent years,  inference methods for assessing the extremal dependence have been in increasing demand. This is especially due to growing requests for multivariate analyses of extreme values in the fields of environmental and economic sciences. The dimension of the random vector under study is often greater than two. For example, 
Figure \ref{fig:map} displays a map of clusters containing seven weather stations in France each; see \citeN{bernard13} for details on the construction of the clusters.
\begin{figure}[h]
\centering
\includegraphics[width=0.50\textwidth]{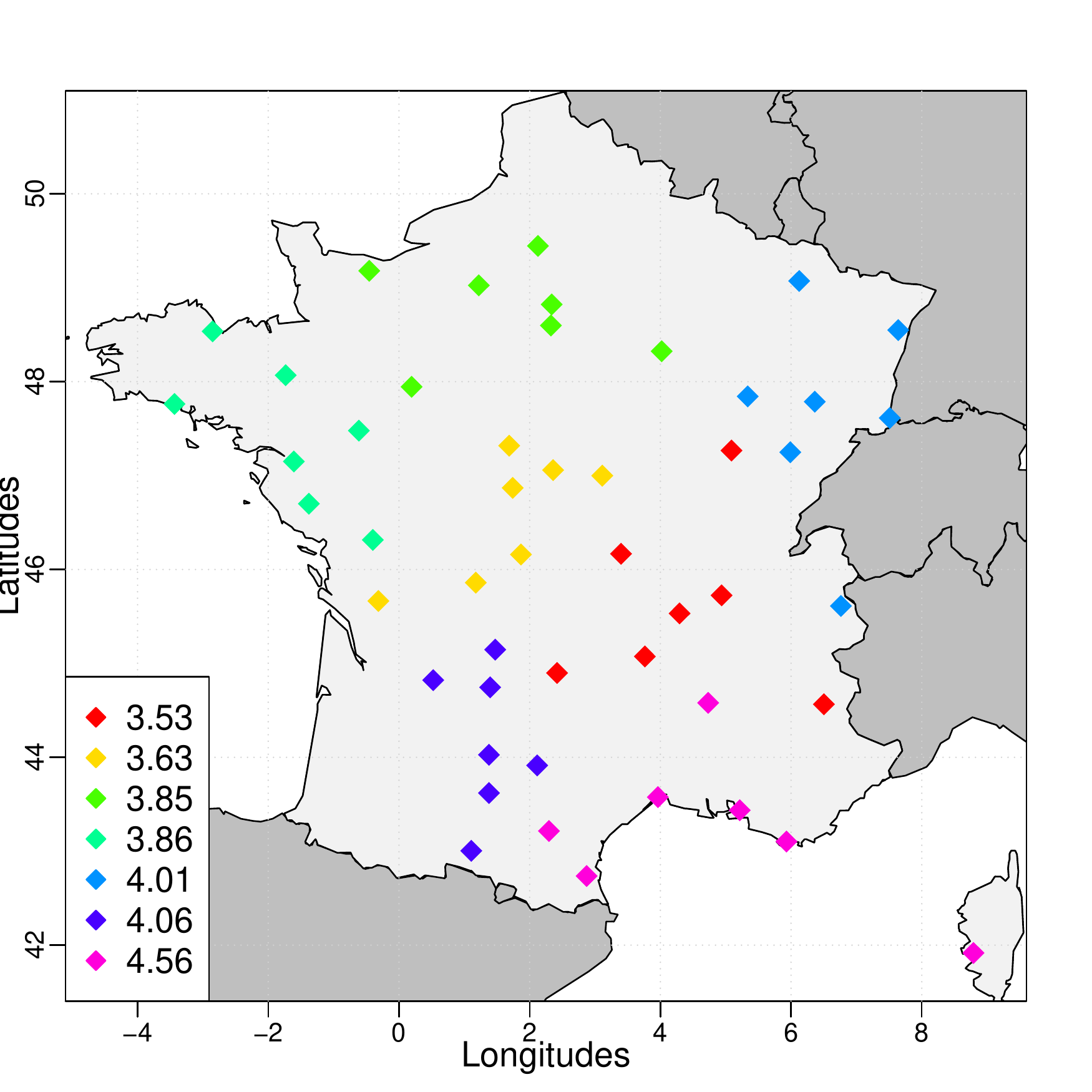}
\caption{
Analysis of French weekly precipitation maxima in the period 1993--2011. 
Clusters  of 49 weather stations   and 
their  estimated extremal coefficients  in dimension $d=7$ obtained with the projected version of the madogram estimator, see Section \ref{sec:data} for details.
}
\label{fig:map}
\end{figure}
The data consist of weekly maxima of hourly rainfall recorded at each station\footnote{Data provided by M\'et\'eo--France and published within the \textsf{R} package \textsf{ClusterMax}, freely available from the homepage of Philippe Naveau, \url{http://www.lsce.ipsl.fr/Pisp/philippe.naveau/}.}. It would be of interest to hydrologists to infer the dependence within each of the seven-dimensional vectors of component-wise maxima and to compare the dependence structures among clusters. Such an endeavor represents the main motivation of this work. 

Let $\bX=(X_1,\ldots,X_d)$ be a $d$-dimensional random vector of maxima 
that follows a multivariate max-stable distribution $G$; for more background on univariate and multivariate extreme-value theory, see for instance \shortciteN{beirlant+g+s+t04}, \citeN{dehaan+f06}, or \citeN{falk+h+r10}. The margins of $G$,
denoted by $F_i(x)=\prob\{X_i\leq x\}$, for all $x\in \real$ and $i=1,\ldots,d$, are univariate max-stable distributions. The joint distribution takes the form
%
%
\begin{equation}
\label{eq:G(x)}
  G(\bx) = C \bigl( F_1(x_1), \ldots, F_d(x_d) \bigr), \qquad \bx \in \real^d,
\end{equation}
where $C$ is an extreme-value copula:
\begin{equation}
\label{eq:ev_copula}
  C(u_1, \ldots, u_d) = \exp \bigl( - \ell( - \log u_1, \ldots, - \log u_d) \bigr), 
  \qquad \bu \in (0, 1]^d,
\end{equation}
with $\ell : [0, \infty)^d \to [0, \infty)$ the so-called stable tail dependence function. The latter function is homogeneous of order one and is therefore determined by its restriction on the unit simplex, the restriction itself being called the Pickands dependence function, denoted here by $A$. Formally, we have
\begin{equation}
\label{eq:ellA}
  \ell( \bz ) = (z_1 + \cdots + z_d) \, A( \bw ), \qquad \bz \in [0, \infty)^d,
\end{equation}
where $w_i = z_i / (z_1 + \cdots + z_d)$ for $i = 1, \ldots, d-1$ and $w_d = 1 - w_1 - \cdots - w_{d-1}$. We view $A$ as a function defined on the $(d-1)$-dimensional unit simplex
\begin{equation}
\label{eq:simplex}
\simp := \left\{ (w_1,\ldots, w_{ d-1}) \in [0,1]^{ d-1}: \sum_{i=1}^{ d-1} w_i \leq 1 \right\}.
\end{equation}

Let $\spA$ be the family of functions $A: \simp \rightarrow [1/d,1]$ that satisfy the following conditions:
\begin{enumerate}
\item[(C1)] $A(\bw)$ is convex, i.e., $A(a\bw_1+(1-a)\bw_2)\leq aA(\bw_1)+(1-a)A(\bw_2)$,
for $a\in[0,1]$ and $\bw_1,\bw_2\in \simp$;
\item[(C2)]  $A(\bw)$ has lower and upper bounds
$$
1/d\leq \max\left(w_1,\ldots,w_{ d-1},w_d \right) \leq A(\bw) \leq 1,
$$
for any $\bw = (w_1, \ldots, w_{ d-1}) \in \simp$ with $w_d=1-w_1-\ldots-w_{d-1}$;  
\end{enumerate}
Any Pickands dependence function belongs to the class $\spA$ (\citeANP{falk+h+r10},  \citeyearNP{falk+h+r10}, Ch.\ 4). The converse is not true, however; see \shortciteANP{beirlant+g+s+t04} (\citeyearNP{beirlant+g+s+t04}, p.\ 257) for a counterexample. A characterization of the class of stable tail dependence functions has been given in \citeN{ressel2013}. 
In condition (C2),  the lower and upper bounds represent  the cases of complete dependence and independence, respectively.
%

Many parametric models have been introduced for modelling 
the extremal dependence for a variety of applications, with summaries to be found in \citeN{kotz2000}  and \citeN{padoan13}.
However, such finite-dimensional parametric models can never cover the full class of Pickands dependence functions.
For this reason, several nonparametric estimators of the Pickands dependence function have been proposed: see for instance \citeN{pickands81}, \shortciteN{cap+f+g97}, \citeN{hall+t00}, \shortciteN{zhang+w+p08},
\citeN{genest2009rank}, \shortciteN{bucher2011},
\citeANP{gudend+s11} (\citeyearNP{gudend+s11}, \citeyearNP{gudend+s12}), and \shortciteN{berghaus2013}. All of these estimators require further adjustments to ensure they are genuine Pickands dependence functions.
%
 
Given an independent random sample from a multivariate distribution with continuous margins and whose copula is an extreme-value copula, we propose a nonparametric estimator for its Pickands dependence function. In the bivariate case, a fast-to-compute and easy-to-interpret estimator based on a type of  
 madogram was introduced by
\shortciteN{naveau+g+c+d}. It has two drawbacks, however: it was only defined for the bivariate case and it is not necessarily a Pickands dependence function itself. 
Our first contribution is to propose a new type of madogram in the multivariate setting, see also \shortciteN{Fonseca13}. 
A second contribution is to regularise the estimator by projecting it onto the space $\spA$, imposing the necessary constraints (C1)--(C2). To do so, we make use of Bernstein polynomials. We admit that the resulting estimator still need not be a Pickands dependence function. Still, simulation results show that imposing (C1)--(C2) already greatly improves the estimation accuracy.


Many regularization strategies have already been considered in the literature. In the bivariate case, \citeN{pickands81} suggested the use of the greatest convex minorant.
\shortciteN{smith+t+y90} proposed to modify a pilote estimator using kernel methods, while \citeN{hall+t00} advocated constrained smoothing splines. 
However, as discussed in \shortciteN{fil+g+s08}, the impact of  these adjustments   on the asymptotic properties of the estimator changes from one case to another, while a general result is unknown. 
The projection estimator approach developed in \shortciteN{fil+g+s08} and \citeN{gudend+s12} provides a general framework based on projections of a pilote estimate onto an increasing sequence of finite-dimensional subsets $\spA_k\subseteq \spA$. The approximation space they proposed consists of piecewise linear functions, yielding computational challenges in higher dimensions.

To bypass these computational hurdles, our strategy is to replace piecewise linear functions by Bernstein polynomials (\citeANP{lorentz53}, \citeyearNP{lorentz53}; \citeANP{sauer91}, \citeyearNP{sauer91}).
In virtue of their optimal shape restriction properties (\citeANP{carnicer+p93}, \citeyearNP{carnicer+p93}),  
Bernstein polynomials are suitable for nonparametric curve estimation (e.g. \citeNP{petrone99}; 
\shortciteNP{chang+h+w+y05}) and shape-preserving regression \cite{wang+g12}.
We provide the asymptotic theory for our estimator and we demonstrate its practical use in dimension seven, which seems to be higher than what has been possible hitherto with nonparametric methods.
The estimation uncertainty can be assessed through a resampling procedure.

Throughout  the paper we use the following notation. Given $\spX\subset \real^n$ and  $n\in \nat$, let 
$\ell^{\infty}(\spX)$ denote the spaces of 
bounded real-valued functions on $\spX$.
For $f:\spX \rightarrow \real$, let $\|f\|_{\infty}=\sup_{\bx \in \spX} |f(\bx)|$. The arrows ``$\conAS$'', ``$\conD$'', and ``$\conW$''  denote almost sure convergence, convergence in distribution of random vectors (see \citeNP{vaart98}, Ch.\ 2) and weak convergence of functions in $\ell^{\infty}(\spX)$ (see \citeNP{vaart98}, Ch.\ 18--19), respectively. Let $L^2(\spX)$ denote the Hilbert space of square-integrable functions $f : \spX \to \real$, with $\spX$ equipped with $n$-dimensional Lebesgue measure; the $L^2$-norm is denoted by $\|f\|_2=(\int_{\spX} f^2(\bx) \, \diff\bx)^{1/2}$. 
For analytical reasons, we view the unit simplex $\simp$ as a subset of $\mathbb{R}^{d-1}$, see \eqref{eq:simplex}, although geometrically, it is perhaps more natural to consider it as a subset of $\mathbb{R}^d$. A similar convention applies to our use of the multi-index $\balpha$ in Section~\ref{sec:bern}.

The paper is organised as follows. 
In Section~\ref{sec:nonparest}, we introduce our multivariate nonparametric madogram estimator and we discuss its properties.
In Section~\ref{sec:bern},
we describe the projection method based on the Bernstein polynomials. 
In Section~\ref{sec:num}, we investigate the finite-sample performance of our estimation method by means of Monte Carlo simulations.
Finally, we apply our approach to French weekly maxima of hourly rainfall in Section~\ref{sec:data}. All proofs are deferred to the appendices.

\section{Madogram estimator}
\label{sec:nonparest}
%
Let $\bX$ be a random vector with continuous marginal distribution functions $F_1, \ldots, F_d$ and whose copula $C$ is an extreme-value copula with stable tail dependence function $\ell$ and Pickands dependence function $A$; see above.
Our estimator is based on the sample version of the multivariate madogram, extending \shortciteN{naveau+g+c+d}, see also \shortciteN{Fonseca13}.
\begin{defi}\label{def:multimado}
For $\bw \in \simp$, the multivariate $\bw$-madogram, denoted by $\nu(\bw)$, 
is defined as the expected distance between the componentwise maximum and the componentwise mean of the 
variables $F^{1/w_1}_{1}(X_{1}), \dots, F^{1/w_d}_{d}(X_{d})$, that is,
\begin{equation}\label{eq:multimd}
\nu(\bw) =
\expect \left[
\bigvee_{i=1}^d\left \lbrace F^{1/w_i}_{i}\left(X_{i}\right)  \right\rbrace -
\frac{1}{d}\sum_{i=1}^dF^{1/w_i}_{i}\left(X_{i}\right)
\right].
\end{equation}
For $w_i = 0$ and $0 < u < 1$, we put $u^{1/w_i} = 0$ by convention.
%
\end{defi}
%

\begin{prop}\label{prop:multimado}
If the random vector $\bX$ has continuous margins and extreme-value copula with Pickands dependence function $A$, then, for all $\bw \in \simp$,
\begin{eqnarray}
\nonumber
  \nu(\bw) &=& \frac{A(\bw)}{1 + A(\bw)} - c(\bw), \\
\label{eq:Amd}
  A(\bw) &=& \frac{\nu(\bw) + c(\bw)}{1 - \nu(\bw) - c(\bw)},
\end{eqnarray}
where $c(\bw) = d^{-1} \sum_{i=1}^d w_i / (1 + w_i)$.
%
\end{prop}
%

The madogram can be interpreted as the $L_1$ distance between the maximum and the average of the random variables $F_1^{1/w_1}(X_1), \ldots, F_d^{1/w_d}(X_d)$. If $w_1 = \ldots = w_d = 1/d$, then the $L_1$ distance is zero if and only if all components $F_i(X_i)$ are equal with probability one, that is, in case of complete dependence.

In the bivariate case, Definition~\ref{def:multimado} is slightly different from the one proposed by \shortciteN{naveau+g+c+d}. 
Here, we use the vector  $\big(F^{1/w_1}_1\left(X_1\right),F^{1/w_2}_2\left(X_2\right)\big)$ instead of 
$\big(F^{w_1}_1(X_1),F^{w_2}_2(X_2)\big)$. 
This new version has the advantage that the sample equivalent of \eqref{eq:Amd} will automatically satisfy condition~(C2).

Assume first that the marginal distributions $F_1, \ldots, F_d$ are known; below, we will estimate them by the empirical distribution functions. Equation~(\ref{eq:multimd}) suggests the statistic 
%
\begin{equation}\label{eq:empmado}
\nu_n(\bw)=\frac{1}{n}\sum_{m=1}^n 
\left(
\bigvee_{i=1}^d\left\lbrace F^{1/w_i}_{i}\left(X_{m,i}\right)\right\rbrace -
\frac{1}{d}\sum_{i=1}^d F^{1/w_i}_{i}\left(X_{m,i}\right)
\right).
\end{equation}
The Pickands dependence function can then be estimated through 
\begin{equation}\label{eq:mado}
A_n^{\MD}(\bw)=\frac{\nu_n(\bw)+c(\bw)}{1-\nu_n(\bw)-c(\bw)}, \qquad \bw\in \simp.
\end{equation}
Next, we estimate the unknown marginal distributions $F_1,\ldots,F_d$ by the empirical distribution functions
\begin{equation}\label{eq:empirical}
F_{n,i}(x)=\frac{1}{n}\sum_{m=1}^n \indic(X_{m,i}\leq x),\qquad i =1,\ldots,d,
\end{equation}
where $\indic(E)$ is the indicator function of the event $E$. Replacing $F_i$ by $F_{n,i}$ in Equation~\eqref{eq:empmado} yields our nonparametric estimators $\widehat{\nu}_n$ and $\widehat{A}_n^{\MD}$ of the multivariate madogram and of the Pickands dependence function, respectively: 
\begin{align*}
  \widehat{\nu}_n(\bw) 
  &= \frac{1}{n}\sum_{m=1}^n 
  \left(
    \bigvee_{i=1}^d \left\lbrace F_{n,i}^{1/w_i}(X_{m,i}) \right\rbrace -
    \frac{1}{d}\sum_{i=1}^d F_{n,i}^{1/w_i}(X_{m,i})
  \right), \\
  \widehat{A}_n^{\MD}(\bw)
  &= \frac{\widehat{\nu}_n(\bw) + c(\bw)}{1 - \widehat{\nu}_n(\bw) - c(\bw)}.
\end{align*}
Other estimators of the margins could be inserted as well. However, the use of the empirical distribution functions requires minimal assumptions and yields an estimator for $A$ which is invariant under monotone transformations.

%

The next theorem summarizes the asymptotic properties related to $A_n^{\MD}$ and $\widehat{A}_n^{\MD}$. 
The asymptotic normality requires a smoothness condition on the extreme-value copula $C$, 
see Example~5.3 in \citeN{seger12}.
\begin{cond}\label{cond:smooth}
For every $i\in\{1,\ldots,d\}$, the partial derivative of $C$ with respect to $u_i$ exists and is continuous on the set $\{\bu\in[0,1]^d: 0< u_i<1\}$.
\end{cond}
Let $\CB$ be a $C$-Brownian bridge, that is, a zero-mean Gaussian process on $[0,1]^d$ with continuous sample paths and with covariance function given by
\begin{equation}\label{eq:covariance}
\Cov(\CB(\bu),\CB(\bv))=C(\bu\wedge\bv)-C(\bu) \, C(\bv),\qquad \bu,\bv\in[0,1]^d,
\end{equation}
where the minimum is considered componentwise. Further, provided Condition~\ref{cond:smooth} is satisfied, define the Gaussian process $\copula$ on $[0, 1]^d$ by
\begin{equation}\label{eq:cop_proc}
\copula(\bu)=\CB(\bu)-\sum_{i=1}^d \frac{\partial C}{\partial u_i}(\bu) \, \CB(1,\ldots,1,u_i,1,\ldots,1),\quad \bu\in[0,1]^d.
\end{equation}
\begin{theo}\label{prop:prop_multimado}
Let $\bX_1, \ldots, \bX_n$ be independent and identically distributed random vectors whose common distribution has continuous margins and extreme-value copula $C$ with Pickands dependence function $A$. Then:
\begin{itemize}
\item[a)]
$
\norm{A_n^{\MD} - A }_\infty \conAS 0$ as $n \to \infty
$
and in $\ell^{\infty}(\simp)$, as $n \to \infty$,
\begin{multline*}
\sqrt{n}(A_n^{\MD}-A)\conW \\
\left((1+A(\bw))^2\frac{1}{d}\sum_{i=1}^d \int_0^1 \bigl( \CB(1,\ldots,1,x^{w_i},1,\ldots,1)-\CB(x^{w_1},\ldots,x^{w_d}) \bigr) \, \diff x\right)_{\bw \in \simp};
\end{multline*}
\item[b)]
$
\norm{\widehat{A}_n^{\MD} - A }_\infty \conAS 0$ as $n \to \infty.
$
Moreover, if Condition~\ref{cond:smooth} is satisfied, then, in $\ell^{\infty}(\simp)$, as $n \to \infty$,
\begin{equation*}\label{eq:wc_pick}
\sqrt{n}(\widehat{A}_n^{\MD}-A)\conW \left(-(1+A(\bw))^2\int_0^1 \copula(x^{w_1},\ldots,x^{w_d}) \, \diff x\right)_{\bw \in \simp}.
\end{equation*}
%
\end{itemize}
\end{theo}

The two conditions (C1)--(C2) are not necessarily satisfied by $\widehat{A}_n^{\MD}$. To ensure both conditions, we propose a projection method based on Bernstein polynomials.

\section{Estimation based on Bernstein polynomials}
\label{sec:bern}

\subsection{Bernstein polynomials on the simplex}
\label{subsec:bern} 

Multivariate Bernstein polynomials, defined on a cube or on a simplex, have been widely discussed in mathematics and statistics, see for example \citeN{Ditzian86} and \citeN{petrone04}. Here our  focus is
on approximating a bounded function $f$ on the simplex $\simp$. In the univariate case, the shape  features of the original function are preserved by its  Bernstein approximation. 
For higher dimensions, shape properties like  convexity may no longer be retained.
The Bernstein--B\'{e}zier polynomials (\citeNP{sauer91}) solve this issue and preserve various shape properties (\citeNP{li11}, \citeNP{lai93}). 

Fix the dimension $d \ge 2$. For positive integer $k$, let $\Gamma_k$ be the set of multi-indices $\balpha = (\alpha_1, \ldots, \alpha_{d-1}) \in \{0, 1, \ldots, k\}^{d-1}$ such that $\alpha_1 + \cdots + \alpha_{d-1} \le k$. The cardinality of $\Gamma_k$ is equal to the number of multi-indices $\balpha \in \{0, 1, \ldots, k\}^d$ such that $\alpha_1 + \cdots + \alpha_d = k$; just set $\alpha_d = k - \alpha_1 - \cdots - \alpha_{d-1}$. Replacing each $\alpha_j$ by $\alpha_j + 1$, we find that the number of such multi-indices is also equal to the number of compositions of the integer $k+d$ into $d$ positive integer parts. The number of such compositions is equal to
\begin{equation}
\label{eq:p}
p_k = \binom{k+d-1}{d-1},
\end{equation}
and so is the cardinality of $\Gamma_k$.
%
Define the Bernstein basis polynomial $b_{\balpha}(\,\cdot\,;k)$ on $\simp$ of degree $k$ by
\begin{eqnarray}
\label{eq:bp}
b_{\balpha}( \bw; k) = \binom{k}{\balpha}\bw^{\balpha}, \qquad \bw\in \simp
\end{eqnarray}
where 
%
$$
\binom{k}{\balpha} = \frac{k!}{\alpha_{1}! \ldots \alpha_{d}!},\qquad
\bw^{\balpha} = w_1^{\alpha_1} \cdots w_d^{\alpha_d}.
$$
The $k$-th degree Bernstein polynomial associated to $A$ is defined as
\begin{equation}\label{eq:polyrap}
 B_A( \bw;k) = \sum_{\balpha \in \Gamma_k} A( \balpha/k) b_{\balpha}( \bw; k),
 \qquad \bw \in \simp.
\end{equation}
%
%
\begin{prop}
\label{prop:conv_bapp}
For every $A \in \spA$ and every $k=1,2,\ldots$,
$$
  \sup_{\bw \in \simp} \abs{ B_A(\bw;k) - A(\bw) } \leq \frac{d}{2\sqrt{k}}.
$$
%
%
\end{prop}
The family of Bernstein--B\'ezier polynomials of degree $k$ is defined as the set 
$$
\mathcal{B}_k =
\left\{ \sum_{\balpha \in \Gamma_k} \beta_{\balpha} \, b_{\balpha}( \, \cdot \, ; k ) : \bbeta \in [0,1]^{p_k} \right\}.
$$
For $\bw \in \simp$, let $\bb_k( \bw )$ be the row vector 
$( b_{\balpha}( \bw; k ), \balpha \in \{0, 1, \ldots, k\}^d: \alpha_1 + \cdots + \alpha_d = k)$. In matrix notation, we have 
$\sum_{\balpha \in \Gamma_k} \beta_{\balpha} \, b_{\balpha}( \bw ; k ) = \bb_k( \bw ) \, \bbeta$, where $\bbeta$ is viewed as a column vector.
\subsection{Shape-preserving estimator}\label{subsec:est}

In this section, we describe how to use
Bernstein--B\'{e}zier polynomials to obtain a 
projection estimator (\shortciteNP{fil+g+s08}) that satisfies (C1)--(C2).
Given a pilot estimator, say $\widehat{A}_n$, the idea is to seek approximate solutions to the constrained optimization problem
%
$$
\widetilde{A}_n=\argmin_{A\in\spA} \norm{ \widehat{A}_n-A }_2.
$$
%
There is no closed-form solution to the above 
equation, and so an approximation based on the sieves method 
is explored. 
Consider a sequence $\spA_k\subseteq \spA$
of constrained multivariate Bernstein--B\'{e}zier polynomial families on $\simp$ given by
\begin{equation}
\label{eq:sievespace}
\spA_k = \left\{ \bw \mapsto B(\bw;k) = \bb_k(\bw)\bbeta_k: 
  \bbeta_k \in [0,1]^{p_k} \text{ such that } \bR_k\bbeta_k \geq \br_k \right\}.
\end{equation}
%
Here, $\bR_k=[\bR^{(1)}_k,\bR^{(2)}_k,\bR^{(3)}_k]^\top$ and $\br_k=[\br^{(1)}_k,\br^{(2)}_k,\br^{(3)}_k]^\top$ are a $(q\times p_k)$ full row rank matrix and a $(q\times 1)$ vector respectively such that the constraint $\bR_k\bbeta_k \geq \br_k$ on the coefficient vector $\bbeta_k$ ensures that each member of $\spA$ satisfies (C1)--(C2).
Details for deriving the block matrices
and vectors of constraints are provided below.
\begin{itemize}
\item[\rone)]
A sufficient condition to guarantee that the function $\bw \mapsto B(\bw;k)$ on $\simp$
is convex is that its Hessian matrix be positive semi-definite.
In order to enforce the latter, we resort by applying Theorem 1 in \citeN{lai93}. 
First, for $s\neq r\in\{0,\ldots,d-1\}$
and two vectors $\bv_r$ and $\bv_s$, where $\bv_r=\bzero$ if $r=0$ and $\bv_r=\be_r$ if $r>0$ with 
$\be_r$ the canonical unit vector (same for $\bv_s$),
the directional derivative of $B$ with respect to the direction $\overrightarrow{\bv_r \bv_s}$ is
$$
D_{\bv_s-\bv_r} B(\bw;k)=k \sum_{\balpha\in\Gamma_{k-1}}\Delta_{s,r}\beta_{\balpha}b_{\balpha} (\bw;k-1), \quad \bw \in \simp
$$
where 
$
\Delta_{s,r}\beta_{\balpha}=(\beta_{\balpha+\bv_s}-\beta_{\balpha+\bv_r}).
$
Second, the second directional derivative of $B$ with respect to the directions $\overrightarrow{\bv_r \bv_s}$ and $\overrightarrow{\bv_r \bv_t}$ is 
$$
D'_{\bv_s-\bv_r,\bv_t-\bv_r}B(\bw;k)=k(k-1) \sum_{\balpha\in\Gamma_{k-2}} 
\Delta_{t,r}\Delta_{s,r}\beta_{\balpha}\,b_{\balpha} (\bw;k-2),\quad \bw\in\simp.
$$
%
Then, the Hessian matrix of $B(\bw;k)$, $\bw\in\simp$, is
$
H_{B}=[ D'_{\bv_s,\bv_t} B(\bw;k)]_{s,t \in \{1,\ldots, d-1\},r=0},
$
and it can be written as 
$$
H_{B}=k(k-1)\sum_{\balpha\in\Gamma_{k-2}}\Sigma_\balpha\,b_{\balpha} (\bw;k-2),\quad \bw\in\simp,
$$
where, for all $\balpha\in\Gamma_{k-2}$, $\Sigma_\balpha$ is a symmetric $(d-1) \times (d-1)$ matrix given by
$$
\Sigma_\balpha=
\begin{pmatrix}
\Delta^2_{1,0}\beta_{\balpha}&	\Delta_{1,0}\Delta_{2,0}\beta_{\balpha} & \cdots& \cdots& \Delta_{1,0}\Delta_{d-1,0}\beta_{\balpha}\\
&\Delta^2_{2,0}\beta_{\balpha} & \Delta_{2,0}\Delta_{3,0}\beta_{\balpha} & \cdots& \Delta_{2,0}\Delta_{d-1,0}\beta_{\balpha}\\
&&\vdots&\vdots&\vdots\\
&&&&\Delta^2_{d-1,0}\beta_{\balpha}
\end{pmatrix}.
$$
By the weak diagonal dominance criterion \cite{lai93} in order to guarantee that 
$\Sigma_\balpha$ is positive semi-definite, it is sufficient to check, for all $\balpha\in\Gamma_{k-2}$ 
and $i\in \{1,\ldots,d-1\}$, the conditions 
$$
\Delta_{i,0}^2 \beta_{\balpha} - \sum_{\substack{j\neq i}} |\Delta_{i,0} \Delta_{j,0} \beta_{\balpha}|\geq 0.
$$
Such conditions produce constraints that are more severe than necessary.
The above conditions can be synthesized in matrix form as $\bR^{(1)}_k \bbeta_k \geq \br^{(1)}_k$ where $\bR^{(1)}_k$ is a $(p_{k-2}(d-1)2^{d-2} \times p_k)$ matrix and  $\br^{(1)}_k$ is the corresponding null vector.
For example, with $d=3$ and $k=3$,
$$
\bR^{(1)}_3=
{\scriptsize 
  \left( 
  \begin{array}{rrrrrrrrrr}
   0 & 1 & 0 & 0 & -1 & -1 & 0 & 1 & 0 & 0 \\ 
   2 & -1 & 0 & 0 & -3 & 1 & 0 & 1 & 0 & 0 \\ 
   0 & -1 & 1 & 0 & 1 & -1 & 0 & 0 & 0 & 0 \\ 
   2 & -3 & 1 & 0 & -1 & 1 & 0 & 0 & 0 & 0 \\ 
   0 & 0 & 1 & 0 & 0 & -1 & -1 & 0 & 1 & 0 \\ 
   0 & 2 & -1 & 0 & 0 & -3 & 1 & 0 & 1 & 0 \\ 
   0 & 0 & -1 & 1 & 0 & 1 & -1 & 0 & 0 & 0 \\ 
   0 & 2 & -3 & 1 & 0 & -1 & 1 & 0 & 0 & 0 \\ 
   0 & 0 & 0 & 0 & 0 & 1 & 0 & -1 & -1 & 1 \\ 
   0 & 0 & 0 & 0 & 2 & -1 & 0 & -3 & 1 & 1 \\ 
   0 & 0 & 0 & 0 & 0 & -1 & 1 & 1 & -1 & 0 \\ 
   0 & 0 & 0 & 0 & 2 & -3 & 1 & -1 & 1 & 0
\end{array} 
\right),
}\quad
\br^{(1)}_3=
{\scriptsize 
\begin{pmatrix}
0\\
0\\
0\\
0\\
0\\
0\\
0\\
0\\
0\\
0\\
0\\
0
\end{pmatrix}.
}$$
A consequence of this approach is that 
\item[\rtwo)] 
$B$ satisfies the upper bound condition in (C2) 
if $\beta_{\balpha}=1$ for the set of coefficients
$\left\{\beta_{\balpha}: {\balpha} = \bzero \text{ or }  {\balpha} = k \, \be_i, \, \forall i=1,\dots, d-1 \right\}$. 
Thus, the $(2d\times p_k)$ matrix and $2d$-dimensional vector of restrictions
are equal to
$$
\bR^{(2)}_k=
{\scriptsize 
  \left( 
  \begin{array}{rrrrrrrr}
1&	 0&\cdots&0&\cdots&0&\cdots&0\\
-1&	 0&\cdots&0&\cdots&0&\cdots&0\\
0&	 0&\cdots&1&\cdots&0&\cdots&0\\
0&	 0&\cdots&-1&\cdots&0&\cdots&0\\
\vdots&\vdots&\vdots&\vdots&\vdots&\vdots&\vdots&\vdots\\
0&	 0&\cdots&0&\cdots&1&\cdots&0\\
0&	 0&\cdots&0&\cdots&-1&\cdots&0\\
\end{array} 
\right),
}
\quad
\br^{(2)}_k=
{\scriptsize 
  \left( 
\begin{array}{r}
1\\
-1\\
1\\
-1\\
\vdots\\
1\\
-1
\end{array} 
\right).
}
$$
\item[\rthree)] 
$B$ satisfies the lower bound condition in (C2) if the restrictions R1)-R2) hold
and the following constraints are fulfilled. 
Specifically, for all $(i,j)\in \{0,\ldots,d-1\}^2$, $i\neq j$, the first directional derivatives with respect to 
$\overrightarrow{\bv_i \bv_j}$, evaluated at the vertices of the simplex, are compared with
the first directional derivatives of the planes $z_0=1$, $z_1=w_1$, $z_2=w_2$, $\ldots$, $z_{d}=1-w_1-w_2-\cdots-w_{d-1}$, with respect to the same directions.  
So, it is sufficient to check the conditions 
$$
D_{\bv_i-\bv_j} B(\bv_j;k) \geq-1, \quad \forall\;(i,j) \in \{0,\ldots,d-1\}^2,\,i\neq j.
$$
As a consequence, it is sufficient to check the conditions
$
\beta_{\balpha} > 1 - 1/k
$ 
for the set of coefficients 
$
\{\beta_{\balpha}: \balpha = \be_i \text{ or }  \balpha = (k-1) \be_i  
\text{ or } \balpha = (k-1) \be_i + \be_j, \; \forall j\neq i=1,\dots, d-1 \}.
$
This can be synthesized in matrix form as $\bR^{(3)}_k \bbeta_k \geq \br^{(3)}_k$ where $\bR^{(3)}_k$ is a 
$(d(d-1)\times p_k)$ matrix and  $\br^{(3)}_k$ is the corresponding vector of $1-1/k$ vaules.
 For example, when $d=3$ and  $k=3$, the constraint matrix is the following: 
$$
\bR^{(3)}_3=
{\scriptsize 
\left( 
	\begin{array}{rrrrrrrrrr}
	0 & 1 & 0 & 0 & 0 & 0 & 0 & 0 & 0 & 0 \\ 
	0 & 0 & 0 & 0 & 1 & 0 & 0 & 0 & 0 & 0 \\ 
	0 & 0 & 1 & 0 & 0 & 0 & 0 & 0 & 0 & 0 \\ 
	0 & 0 & 0 & 0 & 0 & 0 & 0 & 1 & 0 & 0 \\ 
	0 & 0 & 0 & 0 & 0 & 0 & 1 & 0 & 0 & 0 \\ 
	0 & 0 & 0 & 0 & 0 & 0 & 0 & 0 & 1 & 0 \\  
   \end{array} 
\right),
}\quad
\br^{(3)}_3=
{\scriptsize 
	\begin{pmatrix}
	1-1/k\\
	1-1/k\\
	1-1/k\\
	1-1/k\\
	1-1/k\\
	1-1/k\\
	\end{pmatrix}.
}$$
\end{itemize}
The use of the third restriction is justified by the following result.
\begin{prop}\label{prop:lower}
Let $B_A$ be the polynomial \eqref{eq:polyrap}. Assume that $B_A$ is convex on the simplex
and $B_A(\bv_j;k)=1$ for all $j\in\{0,\ldots,d-1\}$. Then, for all $\bw\in\simp$
$$
B_A(\bw;k)\geq\max(w_1,\ldots,w_d) \quad \iff \quad
D_{\bv_i-\bv_j}B_A(\bv_j;k)\geq -1,
$$
for all $(i,j)\in\{0,\ldots,d-1\}^2$, $i\neq j$.
\end{prop}
Recall that the approximate projection estimator of $A$ based on a pilot estimator $\widehat{A}_n$ is given by the solution to the optimization problem
\begin{equation}
\label{eq:approx}
\widetilde{A}_{n,k}=\argmin_{B \in\spA_k} \norm{ \widehat{A}_n-B}_2.
\end{equation}
In case the pilot estimator is the madogram estimator $\widehat{A}_n^{\MD}$, the corresponding projection estimator is denoted by $\widetilde{A}_{n,k}^{\MD}$.

In practice, the estimator $\widetilde{A}_{n,k}$ is evaluated on a finite set of points 
$\{\bw_q : q=1,\ldots,Q\}$, with $Q\in \nat$ and $\bw_q\in\simp$. 
The discretized version of the above solution is given by
%
\begin{equation}\label{eq:BP-MD}
\widetilde{A}_{n,k}(\bw_q)=\bb_k(\bw_q)\widehat{\bbeta}_k,\quad \bw_q\in\simp,\quad q=1,\ldots,Q,
\end{equation}
where $\widehat{\bbeta}_k$ is the minimizer of the constrained least-squares problem
$$
\widehat{\bbeta}_k=\argmin_{\bbeta_k \in [0,1]^{p_k} : \bR_k\bbeta_k \geq \br_k}
\frac{1}{Q}\sum_{q=1}^Q \bigl( \bb_k(\bw_q)\bbeta_k-\widehat{A}_n(\bw_q) \bigr)^2.
$$
This is a quadratic programming problem, whose solution is
\begin{equation}\label{eq:beta}
  \widehat{\bbeta}_k 
  = \bbeta'_k - (\bb_k^\top \bb_k)^{-1} \br_k^\top \bgamma,
\end{equation}
where $\bgamma$ is a vector of Lagrange multipliers and $\bbeta'_k = (\bb^\top_k \bb)^{-1} \bb^\top \widehat{A}_n$ is the unconstrained least squares estimator. The vectors
$\widehat{\bbeta}_k$ and $\bgamma$ can be efficiently computed with an iterative quadratic programming algorithm (e.g. \citeNP{goldfarb+i83}).
A high resolution of \eqref{eq:BP-MD} is obtained with increasing values of $Q$.
Numerical experiments showed that a close approximation of the true Pickands dependence function is already reached with moderate values of $Q$.
However, $Q$ should not be seen as an additional parameter of the projection estimator. The solution \eqref{eq:BP-MD} provides better approximations of the true Pickands dependence function for increasing sample sizes $n$ and polynomial degrees $k$.
%

In order to state the asymptotic distribution of the projection estimator, the following result is required.
\begin{prop}\label{prop:nested}
$\spA_k$, $k=1,2,\ldots$ is a nested sequence in $\spA$.
Furthermore, if $A\in\spA$ satisfies the condition 
\begin{equation}\label{eq:wddA}
\Delta_{i,0}^2A(\balpha/k)-\sum_{j\neq i} |\Delta_{i,0}\Delta_{j,0}A(\balpha/k)|\geq 0, \quad \forall\,k, \balpha \in \Gamma_{k-2}, i\in\{1,\ldots,d-1\}, 
\end{equation}
then there exist polynomials $A_k \in \spA_k$ such that $\lim_{k \to \infty} \sup_{\bw \in \simp} \abs{ A_k( \bw ) - A( \bw ) } =~0$.
%
\end{prop}
The asymptotic distribution of the Bernstein projection estimator based on our multivariate madogram estimator $\widehat{A}_n^{\MD}$ is established in the following proposition.
\begin{prop}\label{prop:prop_bernproj}
Assume that the polynomial's degree, $k_n$, increases with the sample size $n$ in such a way that $k_n/n\rightarrow \infty$ as $n\rightarrow \infty$.
If the Pickands dependence function $A$ 
satisfies the condition \eqref{eq:wddA}, then, for some Gaussian process $Z$,
$$
\sqrt{n}(\widetilde{A}_{n,k_n}^{\MD}-A) \conW \argmin_{Z'\in T_\spA(A)}\|Z'-Z\|_2,\quad n\rightarrow \infty,
$$
in $L^2(\simp)$, where $T_\spA(A)$ is the tangent cone of $\spA$ at $A$, given by the set of limits of all the sequences $a_n(A_n-A)$, $a_n\geq0$ and $A_n\in \spA$.
\end{prop}


It remains an open problem to establish the asymptotic behaviour of the projection estimator 
without condition~\eqref{eq:wddA}. Moreover, if the Pickands dependence function is sufficiently smooth, we conjecture that the approximation rate in Proposition~\ref{prop:conv_bapp} can be improved, leading a slower growth rate needed for the degree of the Bernstein polynomial in Proposition~\ref{prop:prop_bernproj}. The simulation results in Section~\ref{sec:num} confirm that polynomial degrees $k$ much lower than $n$ are already sufficient to achieve good results. 
%

Finally, note that Proposition~\ref{prop:prop_bernproj} and, in fact, everything else in this section applies to any estimator of the Pickands dependence function which satisfies a suitable functional central limit theorem. We are grateful to an anonymous Referee for having pointed this out.

\subsection{Confidence bands}\label{subsec:comp}
%
%
We construct confidence bands using a resampling method. For $\bw\in\simp$
and $0<\tilde{\alpha}<1$, the bootstrap $(1-\tilde{\alpha})$ pointwise confidence band, based on the estimates 
$\widetilde{A}^{*(r)}_{n,k}(\bw)$, $r=1,2,\ldots$, obtained from the bootstrapped sample $\bX^{(r)}_n=(\bX^{(r)}_1,\ldots,\bX^{(r)}_d)$,
has the drawback that the lower and upper limits of the band are rarely  convex and continuous.
To bypass this hurdle, we followed the strategy to work with the estimated  Bernstein polynomials' coefficients themselves. 
Specifically, let  $\widehat{\bbeta}^{*(r)}_k$ be the  Bernstein polynomials' coefficient estimator based on the bootstrap sample $\bX_n^{(r)}$, $r=1,2,\ldots$, we define a bootstrap simultaneous $(1-\tilde{\alpha})$ confidence band
specifying the lower $\widetilde{A}^{L}_{n,k}(\bw)$ and upper $\widetilde{A}^{U}_{n,k}(\bw)$ limits 
as
\begin{equation}
\label{eq:confidentIntervals}
\left[\sum_{\balpha \in \Gamma_k} \widehat{\beta}^{*\lceil r(\tilde{\alpha}/2) \rceil}_{\balpha} b_{\balpha} (\bw;k);
\;
\sum_{\balpha \in \Gamma_k} \widehat{\beta}^{*\lceil r(1-\tilde{\alpha}/2)\rceil}_{\balpha} b_{\balpha} (\bw;k)
\right],
\quad \bw\in\simp,
\end{equation}
where $\widehat{\beta}^{*\lceil r (\tilde{\alpha}/2)\rceil}_{\balpha}$ and
$\widehat{\beta}^{*\lceil r(1-\tilde{\alpha}/2)\rceil}_{\balpha}$,
for all $\balpha\in\Gamma_k$,
correspond to the
$\lceil r(\tilde{\alpha}/2)\rceil$ and $\lceil r(1-\tilde{\alpha}/2)\rceil$ ordered statistics
respectively and $b_{\balpha} (\bw;k)$ is the Bernstein basis polynomial of degree $k$, see \eqref{eq:bp}.
Although this approach does not guarantee convex confidence bands, 
it works very well in our simulations, where we find
that the convexity is violated only when dependence is weak.
Another possibility, that can be considered, is to bootstrap bands for unconstrained estimators and then apply projection to the lower and upper bound, as pointed out by an anonymous Referee.  Our specific simulations results indicate that our method 
performs slightly better than this valuable alternative.

\section{Simulations}\label{sec:num}

To visually illustrate the gain  in  implementing our  Bernstein-B\'ezier projection approach, 
Figure \ref{fig:comparison} compares the madogram (MD) estimator $\widehat{A}^{\MD}_{n}$ defined by (\ref{eq:mado}) 
with its  Bernstein-B\'ezier-projection (BP) version defined by (\ref{eq:BP-MD}) for the special case of  the symmetric logistic model (SL, \citeNP{tawn90}) with $d=3$ and  $\alpha' = 0.3$.  
For all sample sizes ($n=20,50,100$), an improvement can be observed by comparing the estimated contour lines (dotted) in  the top and bottom panels. 
This is particularly true for  a small sample size like $n=20$, the corrected version providing smoother  and more realistic contour lines.
\begin{figure}[h!]

\emph{\hspace{40mm}  $n=20$  \hspace{30mm} $n=50$   \hspace{30mm} $n=100$}\\
\begin{minipage}{.10\textwidth}
     \begin{center}
     	\emph{MD}
     \end{center}
\end{minipage}
\begin{minipage}{.90\textwidth}
	\centering
	\includegraphics[width=0.30\textwidth]{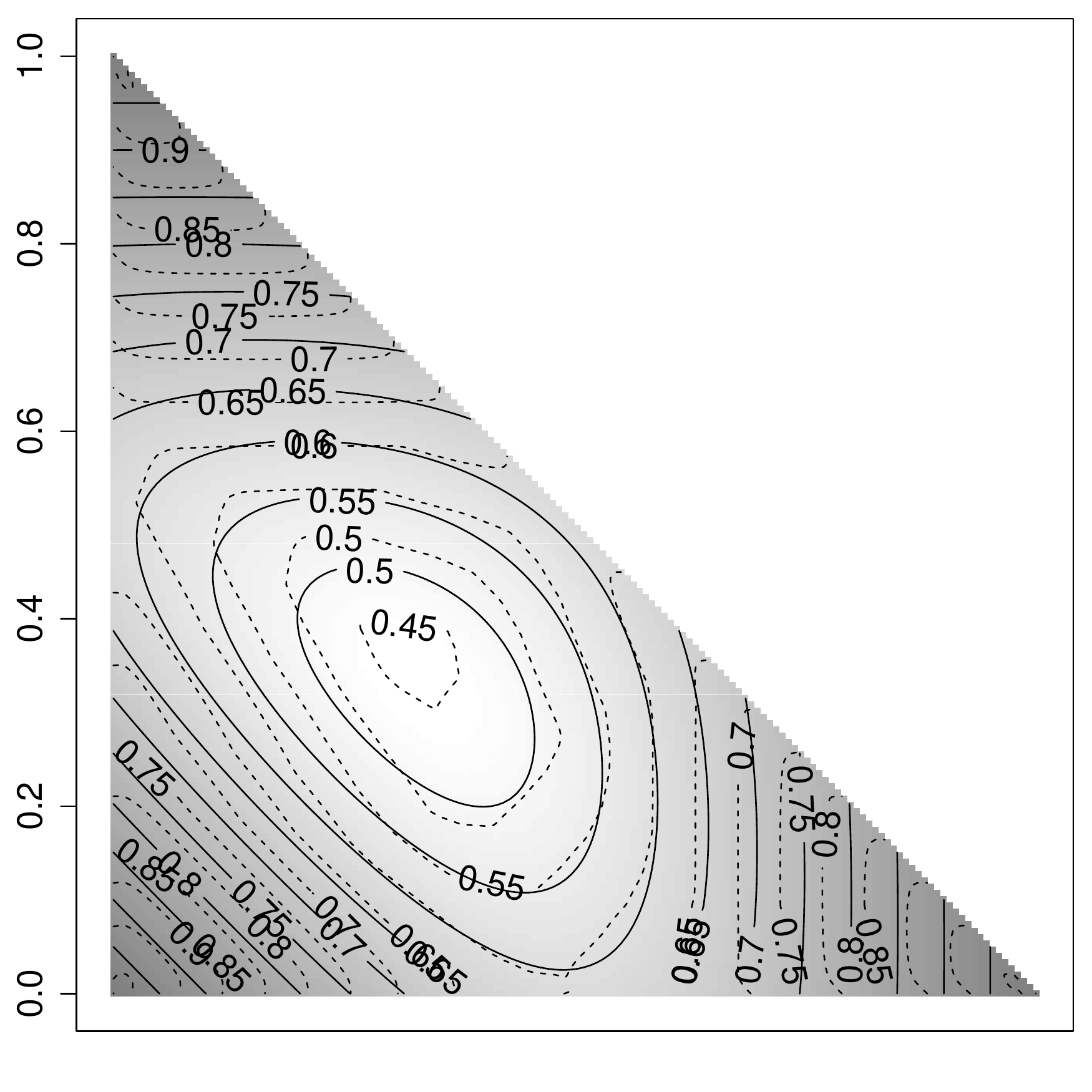}
	\includegraphics[width=0.30\textwidth]{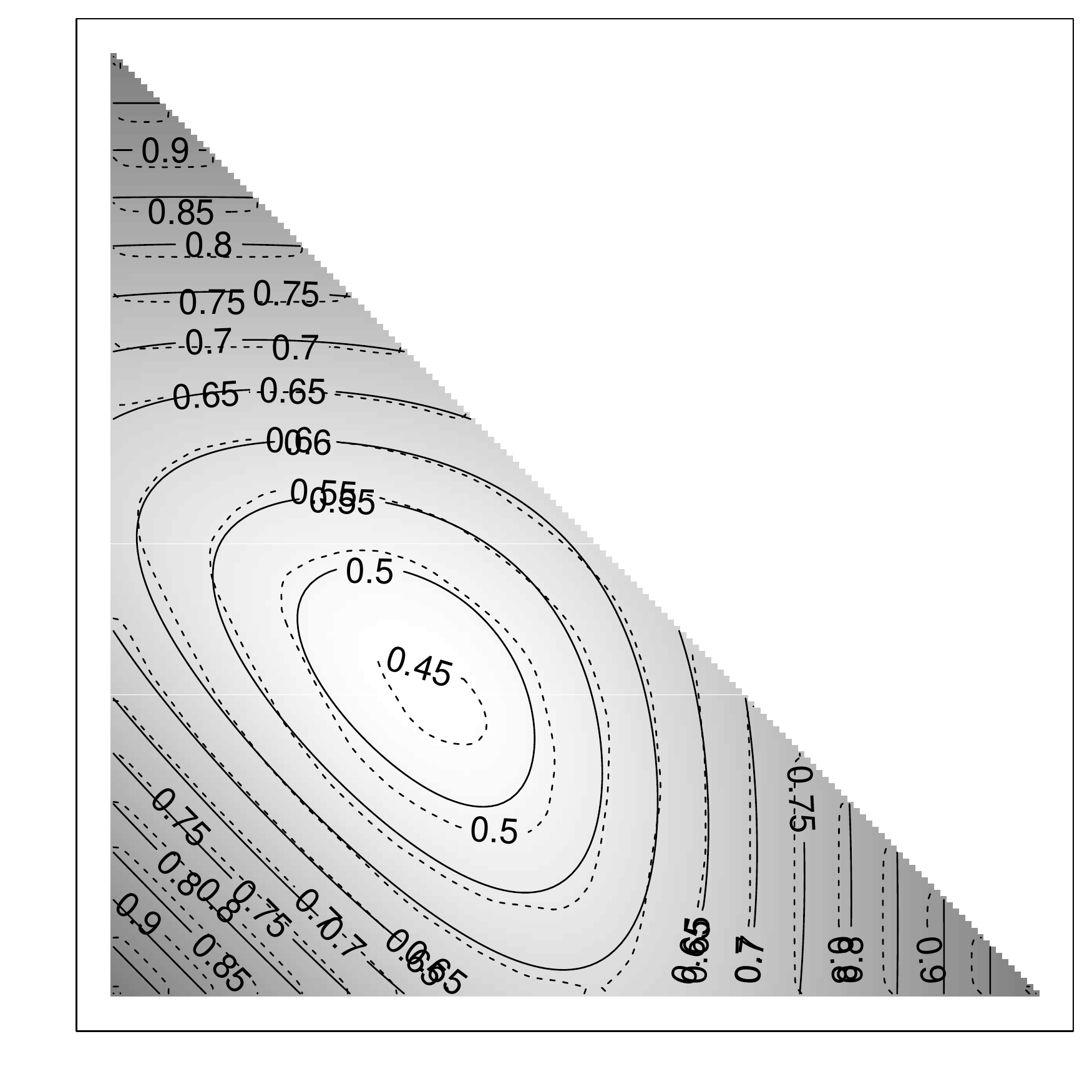}
	\includegraphics[width=0.30\textwidth]{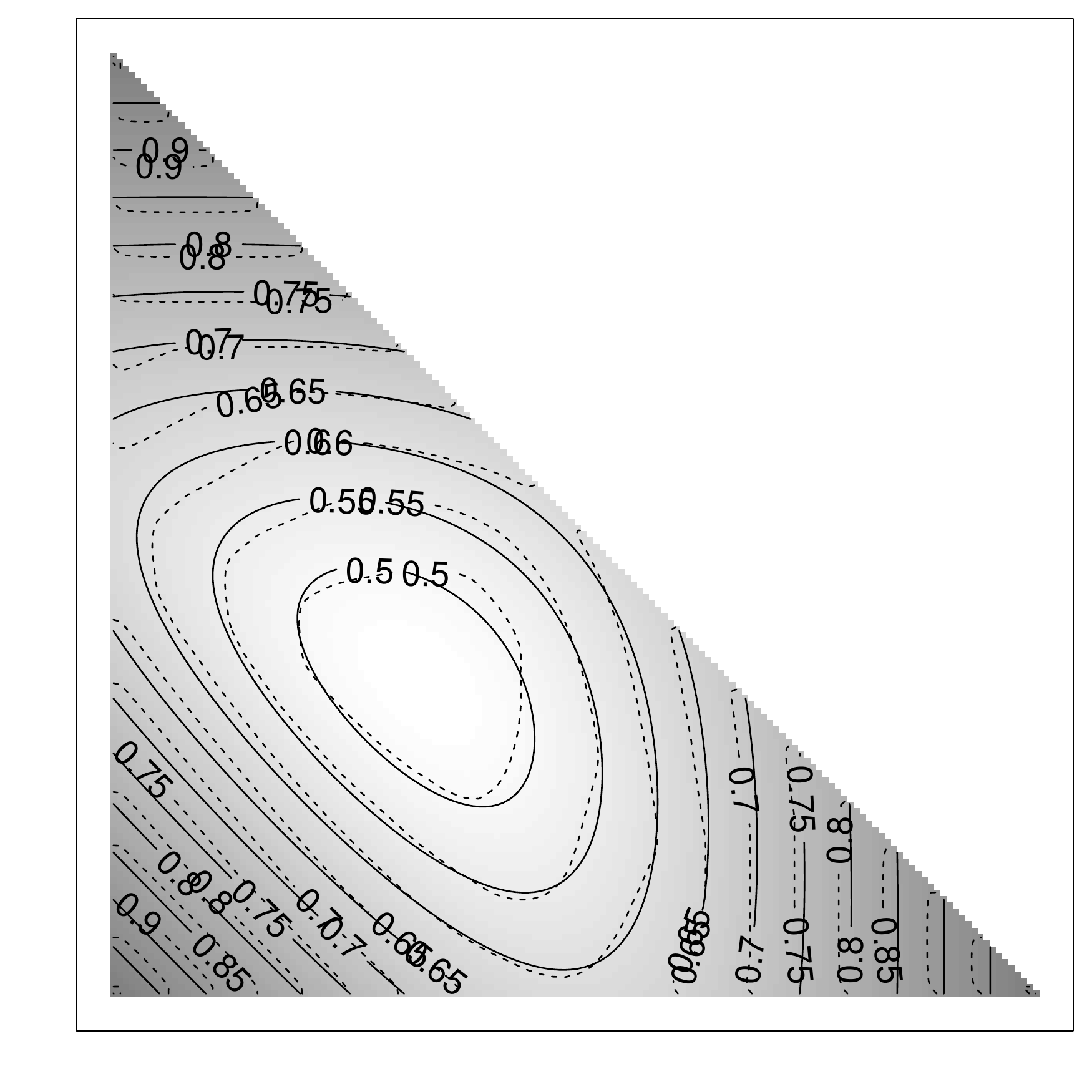}\\
\end{minipage}
\begin{minipage}{.10\textwidth}
\begin{center}
		\emph{BP-MD}
\end{center}
\end{minipage}
\begin{minipage}{.90\textwidth}
	\centering
	\includegraphics[width=0.30\textwidth]{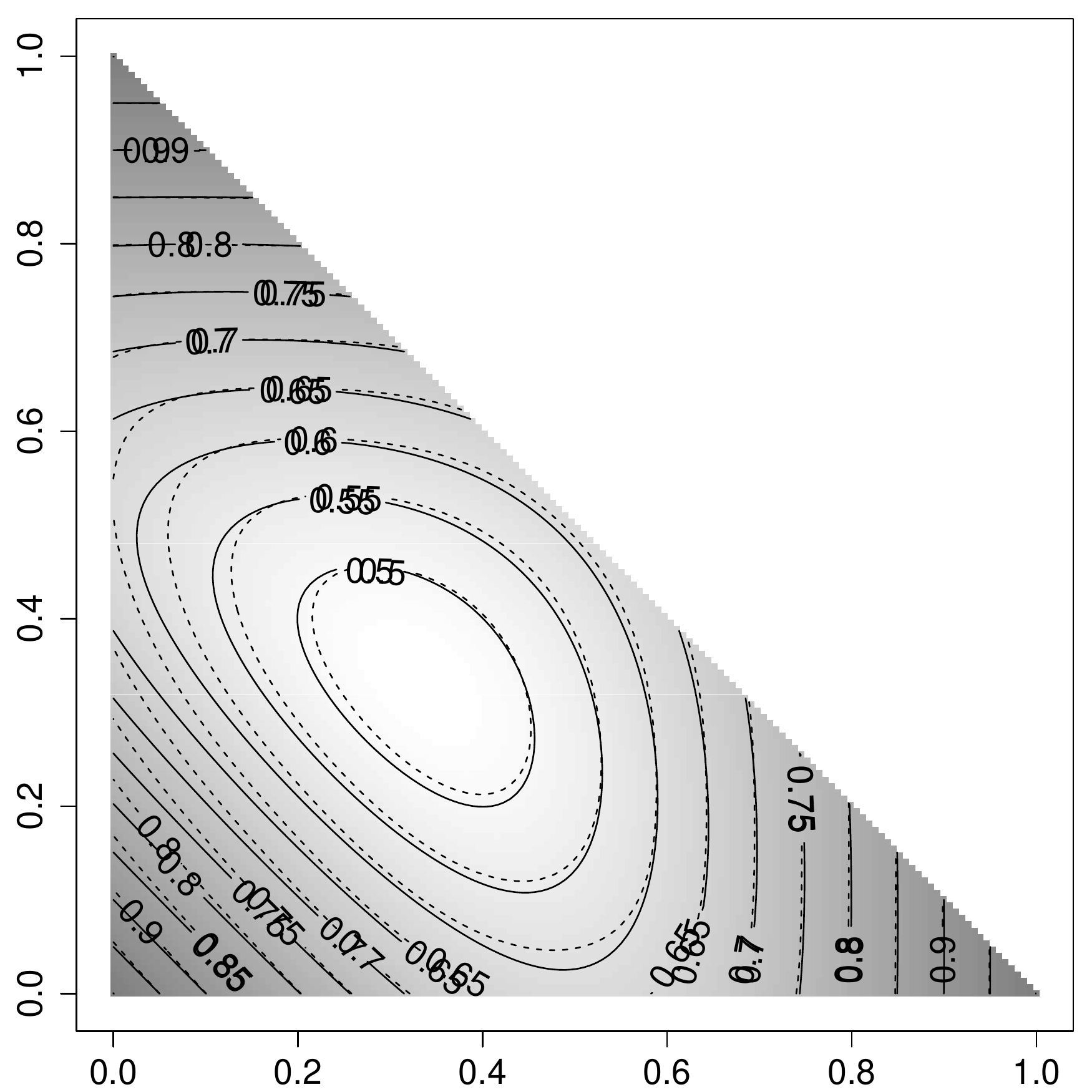}
	\includegraphics[width=0.30\textwidth]{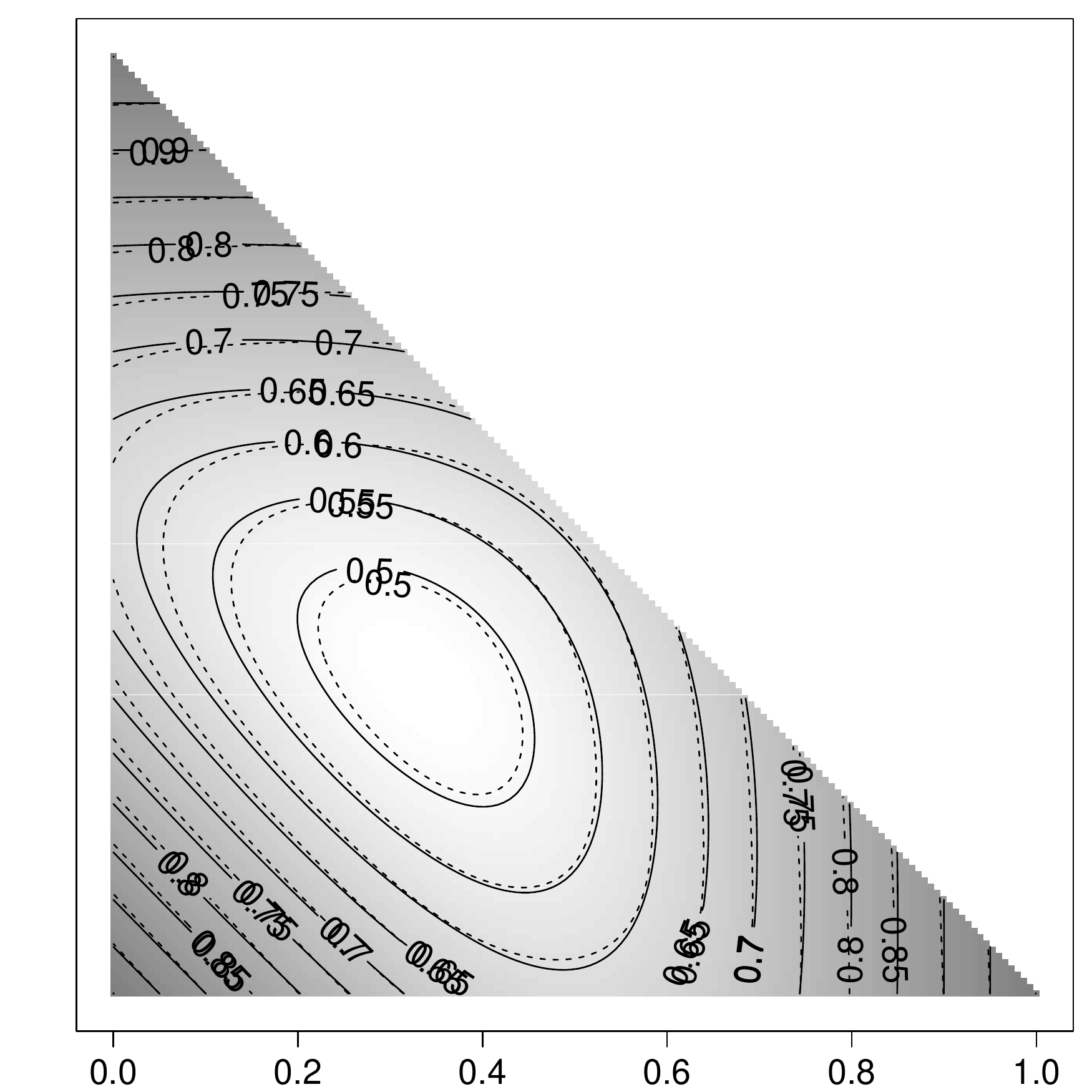}
	\includegraphics[width=0.30\textwidth]{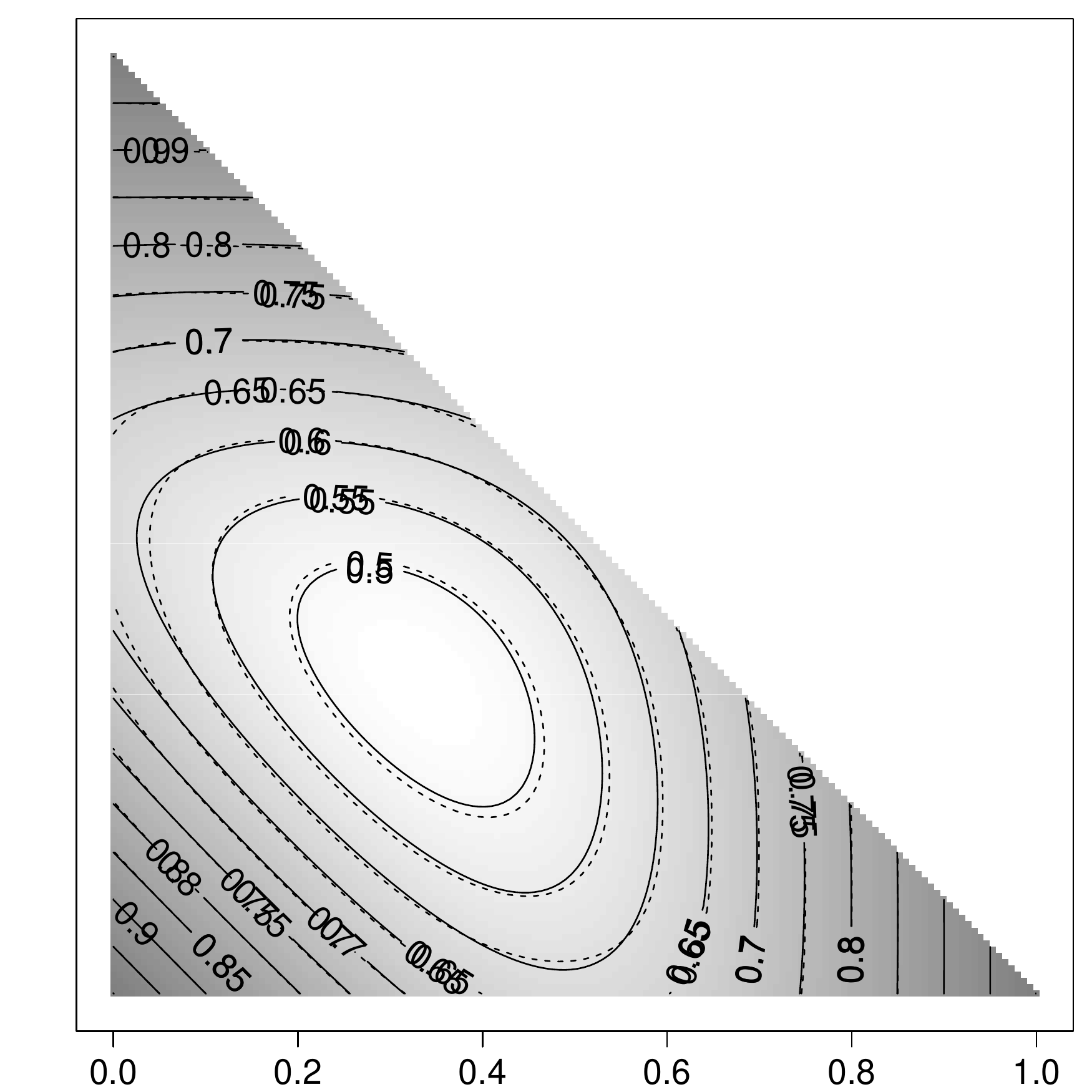}\\
\end{minipage}
\caption{
Estimates (dashed lines) of the Pickands dependence function obtained with the MD estimator (top row) and  
its BP version (bottom row) with polynomial degree $k=14$. 
The solid line is the true Pickands dependence function.
Each column represents a different sample size.
}
\label{fig:comparison}
\end{figure}

To guarantee a good approximation of $A$  with $\widetilde{A}_{n,k}$, Proposition \ref{prop:prop_bernproj} suggested  to set  a large polynomial degree $k$ for large sample sizes, see also \shortciteN{fil+g+s08}, \citeN{gudend+s11}, \citeN{gudend+s12}. 
But  computational time limits   restrict  the choice of $k$. 
Figure \ref{fig:comparisonK} explores this issue for the logistic model with $\alpha'=0.3$ and    $n=100$. 
As expected from the theory, the choice of $k$ is not anecdotical. A shift in the contour lines appears for the small value    $k=5$, see the left panel of Figure \ref{fig:comparisonK}.  
This undesirable feature disappears for a moderate value of $k$, see the right panel with $k=14$.
%
\begin{figure}[h!]
\emph{\hspace{40mm}  $k=5$  \hspace{33mm} $k=11$   \hspace{30mm} $k=14$}\\
\begin{minipage}{.10\textwidth}
	\begin{center}
		\emph{BP-MD}
	\end{center}
\end{minipage}
\begin{minipage}{.90\textwidth}
	\centering
	\includegraphics[width=0.30\textwidth]{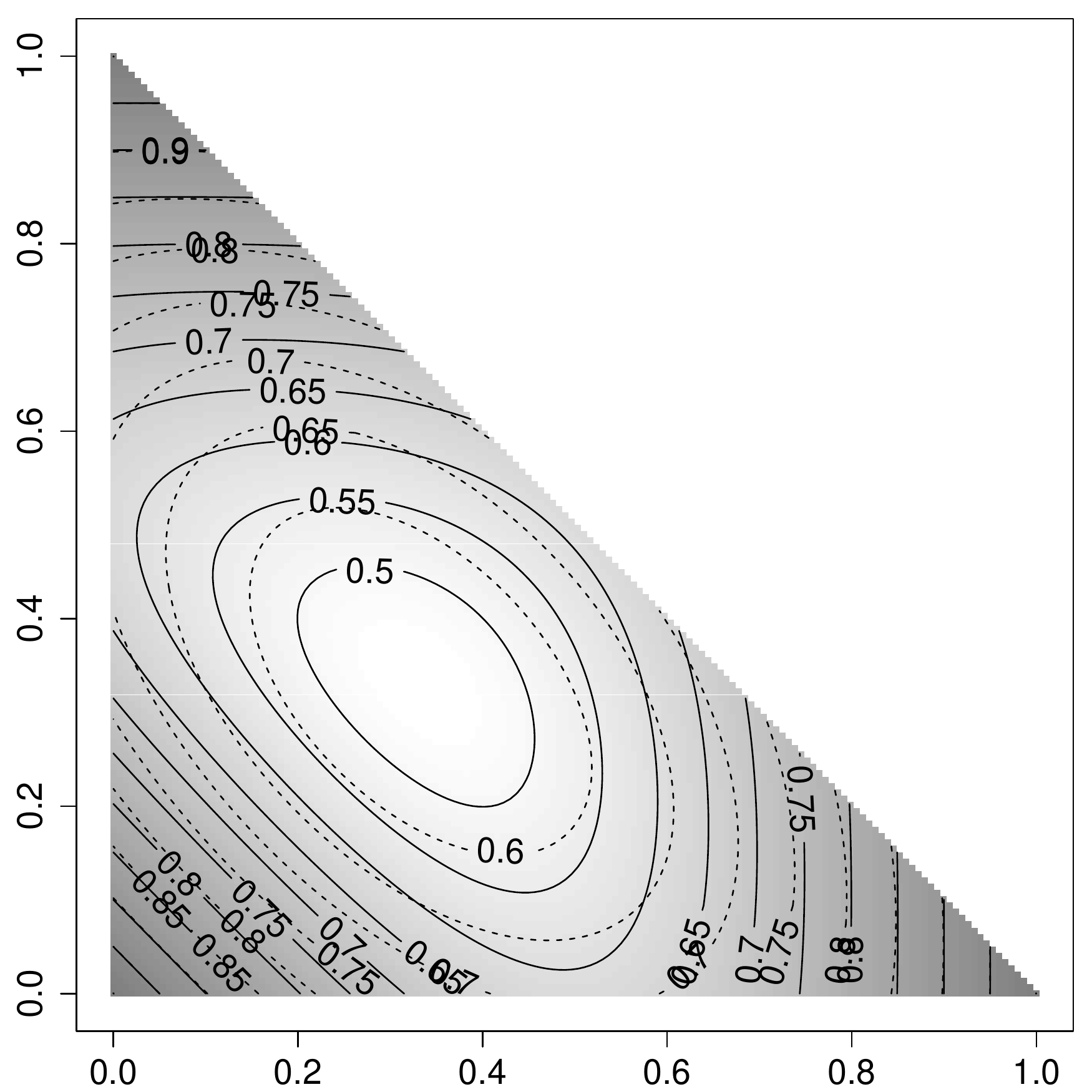}
	\includegraphics[width=0.30\textwidth]{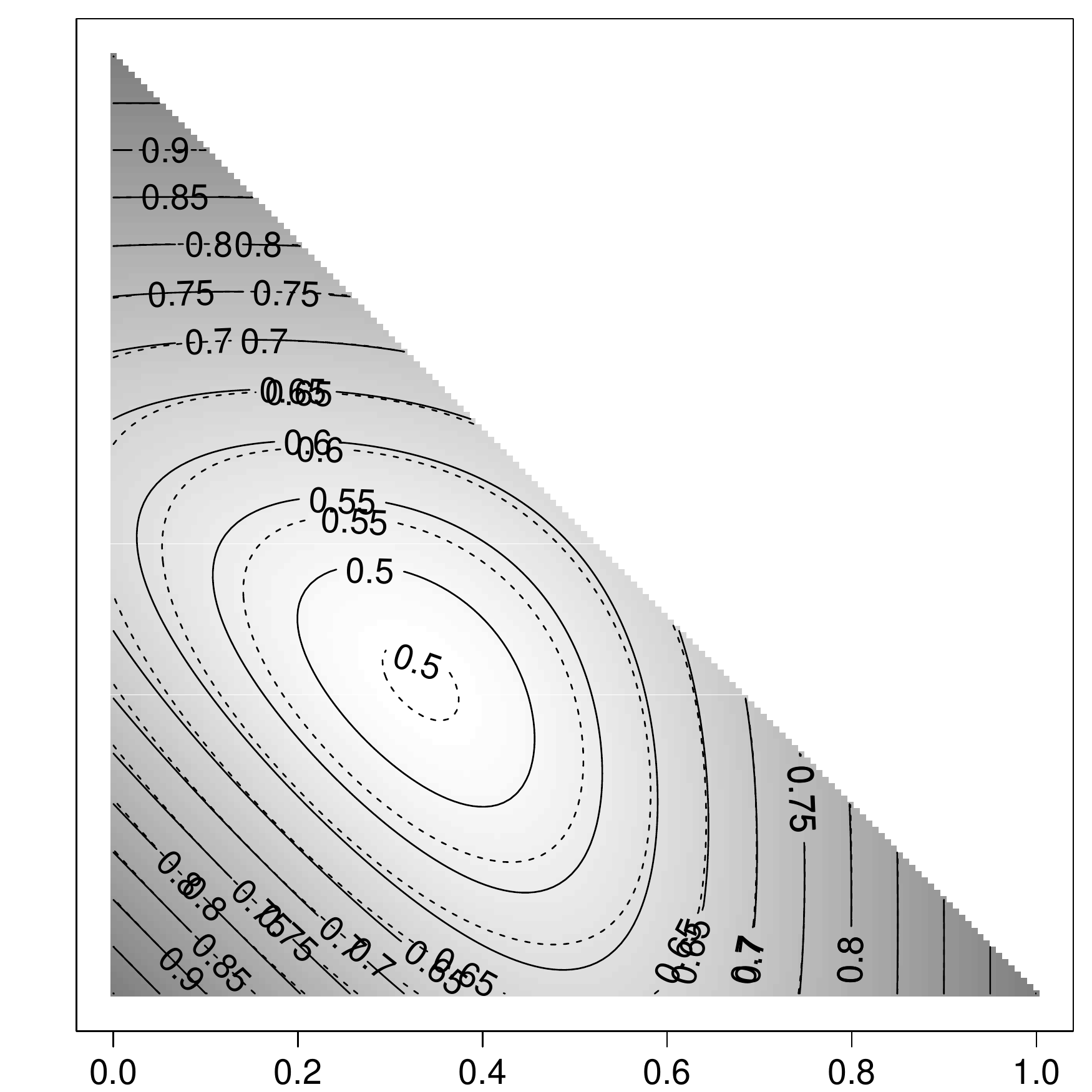}
	\includegraphics[width=0.30\textwidth]{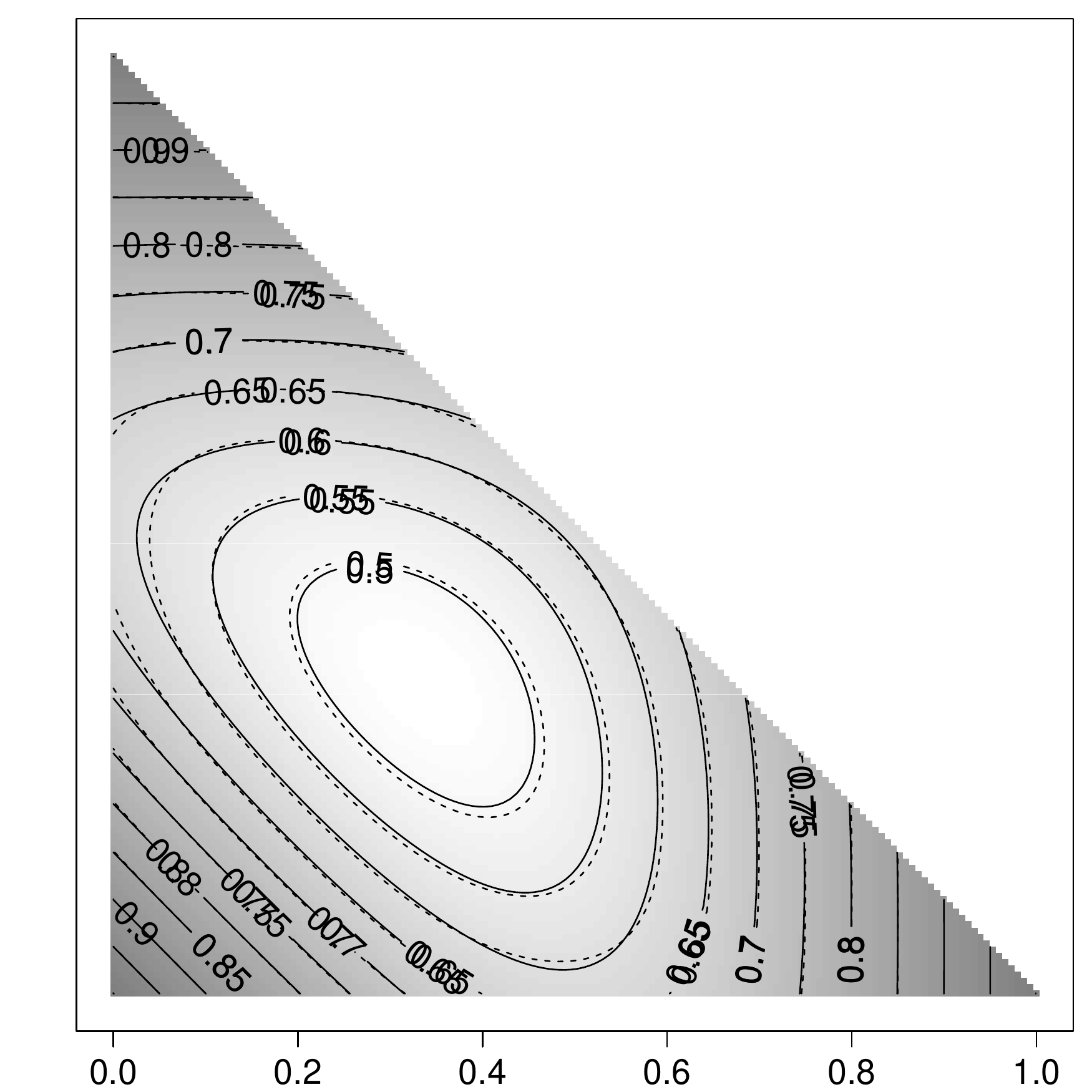}\\
\end{minipage}
\caption{Same as Figure \ref{fig:comparison} but with $n=100$ and 
three different values of the polynomial degree.
}
\label{fig:comparisonK}
\end{figure}

To go beyond these visual checks, we also compute a Monte-Carlo approximation of the mean integrated squared error  
\begin{equation*}
\label{eq:mise}
\mbox{MISE}(\widehat{A}_n,A) = \expect\left\{ \int_{\simp} \left(\widehat{A}_n(\bw)-A(\bw)\right)^2 \diff\bw\right\}, 
\end{equation*}
for a variety of setups.  The approximate MISE is obtained by repeating $1000$ times 
a given inference method for three different sample sizes $n=50,100, 200$.
%
Different dependence strength of the logistic model has been explored setting the parameter $\alpha'$
between $0.3$ (strong dependence) and 1 (independence).
Table \ref{tab:MISE} compares four non-parametric estimators introduced in Section \ref{sec:nonparest}: the madogram estimator (MD), the \citeN{pickands81} estimator (P), the multivariate version of the \citeN{hall+t00} estimator (HT), and finally the multivariate extension of the \citeN{cap+f+g97} estimator.
For comparison purposes we have also considered the weighted and endpoint-corrected versions of the P and CFG estimators as discussed in \citeN{gudend+s12}, denoted by Pw and CFGw respectively. 
%
%
%
\begin{table}[t!]
\begin{center}
{\footnotesize \begin{tabular}{ccccccc}
\toprule
Sample size $n$& Estimator & \multicolumn{5}{c}{Parameter $\alpha'$}\\
& & $0.3$ & $0.5$ & $0.7$ & $0.9$ &$1$\\
\midrule
$50$ 
 & P  & $ 4.25\times 10^{-4} $ & $ 8.06\times 10^{-4} $ & $ 1.47\times 10^{-3} $ & $ 2.45\times 10^{-3} $ & $ 2.50\times 10^{-3} $\\
 & Pw  & $ 1.45\times 10^{-4} $ & $ 5.13\times 10^{-4} $ & $ 1.26\times 10^{-3} $ & $ 2.53\times 10^{-3} $ & $ 2.81\times 10^{-3} $\\
 & CFG & $ 2.36\times 10^{-4} $ & $ 6.92\times 10^{-4} $ & $ 1.87\times 10^{-3} $ & $ 4.07\times 10^{-3} $ & $ 5.02\times 10^{-3} $\\
 & CFGw & $ 9.17\times 10^{-5} $ & $ 4.45\times 10^{-4} $ & $ 1.24\times 10^{-3} $ & $ 2.66\times 10^{-3} $ & $ 3.07\times 10^{-3} $\\
 & HT & $ 2.64\times 10^{-4} $ & $ 8.54\times 10^{-4} $ & $ 2.59\times 10^{-3} $ & $ 5.13\times 10^{-3} $ & $ 5.65\times 10^{-3} $\\
 & MD & $ 1.80\times 10^{-4} $ & $ 8.66\times 10^{-4} $ & $ 1.91\times 10^{-3} $ & $ 3.02\times 10^{-3} $ & $ 2.87\times 10^{-3} $\\
 \midrule
$100$ 
 & P & $ 1.53\times 10^{-4} $ & $ 3.16\times 10^{-4} $ & $ 6.98\times 10^{-4} $ & $ 1.20\times 10^{-3} $ & $ 1.39\times 10^{-3} $\\
 & Pw & $ 6.36\times 10^{-5} $ & $ 2.38\times 10^{-4} $ & $ 6.51\times 10^{-4} $ & $ 1.25\times 10^{-3} $ & $ 1.51\times 10^{-3} $\\
 & CFG & $9.54\times 10^{-5} $ & $3.27\times 10^{-4} $ & $8.66\times 10^{-4} $ & $1.78\times 10^{-3} $ & $2.15\times 10^{-3} $\\
 & CFGw & $4.32\times 10^{-5} $ & $2.21\times 10^{-4} $ & $6.35\times 10^{-4} $ & $1.24\times 10^{-3} $ & $1.39\times 10^{-3} $\\
 & HT & $ 2.61\times 10^{-4} $ & $ 7.66\times 10^{-4} $ & $ 2.16\times 10^{-3} $ & $ 4.24\times 10^{-3} $ & $ 5.27\times 10^{-3} $\\
 & MD & $ 7.02\times 10^{-5} $ & $ 3.18\times 10^{-4} $ & $ 7.91\times 10^{-4} $ & $ 1.19\times 10^{-3} $ & $ 1.09\times 10^{-3} $\\
\midrule
$200$ 
 & P & $ 5.87\times 10^{-5} $ & $1.54\times 10^{-4} $ & $ 3.40\times 10^{-4} $ & $ 6.25\times 10^{-4} $ & $ 7.24\times 10^{-4} $\\
 & Pw & $ 3.01\times 10^{-5} $ & $1.31\times 10^{-4} $ & $ 3.28\times 10^{-4} $ & $ 6.60\times 10^{-4} $ & $ 7.59\times 10^{-4} $\\
 & CFG & $ 3.87\times 10^{-5} $ & $ 1.58\times 10^{-4} $ & $ 4.00\times 10^{-4} $ & $ 8.31\times 10^{-4} $ & $ 8.52\times 10^{-4} $\\
 & CFGw & $ 2.12\times 10^{-5} $ & $ 1.23\times 10^{-4} $ & $ 3.24\times 10^{-4} $ & $ 6.36\times 10^{-4} $ & $ 5.90\times 10^{-4} $\\
 & HT & $ 2.55\times 10^{-4} $ & $ 7.31\times 10^{-4} $ & $ 2.05\times 10^{-3} $ & $ 3.82\times 10^{-3} $ & $ 5.85\times 10^{-3} $\\
 & MD & $ 3.17\times 10^{-5} $ & $ 1.58\times 10^{-4} $ & $ 3.70\times 10^{-4} $ & $ 5.81\times 10^{-4} $ & $ 4.91\times 10^{-4} $\\
\bottomrule
\end{tabular}}
\caption{MISE of four estimators of the Pickands dependence function, and some weighted version, based on a trivariate symmetric
logistic dependence model for different parameter values and sample sizes.} 
\label{tab:MISE}
\end{center}
\end{table}
%
%
We can see that the MD estimator provided the best results if compared with the other classical non-parametric estimators.
Taking into account also the weighted versions, it turns out that the CFGw estimator performs the best,  
especially for small sample sizes ($n=50$). With a medium dependence ($\alpha'=0.5, 0.7$) the estimators
provide similar results. With a weak dependence or in the independence case ($\alpha'=0.9,1$), the MD estimator 
still provides the best results, especially for small and moderate sample sizes ($n=50,100$).
%
%
%
\begin{table}[t!]
\begin{center}
{\footnotesize \begin{tabular}{ccc}
\toprule
& \multicolumn{2}{c}{Projection method}\\
& Bernstein-B\'ezier & Discrete spectral measure\\
\midrule
& \multicolumn{2}{c}{\% Improvement}\\ 
\begin{tabular}{ccc}
&\\
$n$ & $\alpha'$ & $k$\\
\midrule
50 & 0.3 & 23\\
     & 0.5 & 20\\
     & 0.7 & 16\\
     & 0.9 & 6\\
     & 1 & 3\\
\midrule
100 & 0.3 & 23\\
     & 0.5 & 20\\
     & 0.7 & 16\\
     & 0.9 & 6\\
     & 1 & 3\\
\midrule
200 & 0.3 & 23\\
     & 0.5 & 20\\
     & 0.7 & 16\\
     & 0.9 & 6\\
     & 1 & 3\\
\end{tabular}
&
\begin{tabular}{cccccc}
\multicolumn{4}{c}{Estimator}\\
P & Pw& CFG & CFGw & HT & MD\\
\midrule
18.11 & 13.34  & 76.84& 18.97 & 51.53 & 8.50\\
8.19 & 5.44 & 13.98 & 1.46 & 12.52 & 2.22\\
15.60 & 11.01 & 4.43&  2.10 & 9.05 & 6.48\\
44.70 & 25.92 & 3.98 & 6.51 & 16.93 & 48.72\\
69.95 & 34.53 & 4.92 & 9.04 & 34.68 & 93.60\\
\midrule
16.59 & 13.36 & 59.75 & 13.43 & 45.45 & 7.41\\
5.85 & 3.83 & 7.59 & 0.63 & 9.78 & 1.23\\
9.89 & 8.15 & 2.21 & 0.95 & 6.42 & 2.74\\
34.95 & 23.98 & 3.48 & 6.50 & 8.33 & 26.72\\
68.00 & 39.35 & 5.93 & 11.50 & 19.22 & 87.46\\
\midrule
15.16 & 10.63 & 37.73 & 5.66 & 44.72 & 5.05\\
3.06 & 2.51 & 3.80 & 0.41 & 9.06 & 0.13\\
5.70 & 5.22 & 0.90 & 0 & 5.60 & 0.76\\
25.22 & 20.48 & 3.43 & 6.07 & 5.53 & 13.39\\
69.17 & 46.32 & 8.63 & 16.06 & 10.88 & 81.99\\
\end{tabular}
&
\begin{tabular}{cc}
\multicolumn{2}{c}{Estimator}\\
Pw & CFGw\\
\midrule
2.14 & 0.82\\
5.51 & 1.17\\
11.03 & 3.17\\
22.39 & 4.37\\
29.07 & 4.89\\
\midrule
1.27 & 0.40\\
3.52 & 0.84\\
7.51 & 1.37\\
23.13 & 4.07\\
36.10 & 8.11\\
\midrule
0.60 & 0\\
2.10 & 0\\
4.85 & 0.88\\
18.52 & 3.28\\
40.99 & 11.89\\
\end{tabular}
\\
\bottomrule
\end{tabular}}
\caption{Percentage improvement of the MISE gained with the projection method.} 
\label{tab:improve}
\end{center}
\end{table}
%

Table \ref{tab:improve} shows how an initial estimate of the Pickands dependence function improves using the projection method.
The improvement is computed by 
$$
\frac{MISE_N - MISE_P}{MISE_N} \times 100,
$$
and is reported in columns 3--6, 
where $MISE_N$ and $MISE_P$ are the MISE obtained with a non-parametric estimator and its projection respectively. 
As before, MISE provides a Monte-Carlo approximation of $\mbox{MISE}(\widehat{A}_n,A)$ 
obtained with 1000 random samples. The true dependence
structure is still the symmetric logistic model. $\alpha'$ denotes the model parameter, and $n$ and $k$ are the sample size and the
polynomial degree respectively. 
Estimates obtained with the initial non-parametric are regularized
using the BP method. The order of the
polynomial exploited is an ``optimal" value of $k$, that is the $k$ value chosen in
a such way that the MISE does not decrease significantly for larger values of $k$. It turns out that with a weak dependence
a small value of $k$ is enough, conversely with a strong dependence a large value of $k$ is needed.
This makes sense if we view a dependence structure as an added  complexity, especially with respect to the independence case, the simplest possible model. 
In such a framework, the polynomial degree has to be higher to capture this extra information. 
The improvements obtained with the classical estimators, sorted from largest to smallest, are: MD, CFG, P and HT. 
As expected, with Pw and CFGw the improvements are the smallest. For each
estimator, the improvements sorted from largest to smallest, are obtained with: independence ($\alpha'=1$), strong dependence ($\alpha'=0.3$), weak dependence ($\alpha'=0.9$) and medium dependence ($\alpha'=0.5,0.7$). 
These results are compared with those provided in \citeN{gudend+s12} that are obtained
with the discrete spectral measure projection method proposed by the same authors (see columns 7,8).
We can conclude that overall the BP method provides a better percentage improvement. 

To explore the validity of our procedure to derive a bootstrap pointwise and simultaneous $(1-\tilde{\alpha})$ 
confidence band described in Section \ref{subsec:comp},  
Table \ref{tab:covprob} displays $95\%$ coverage probabilities from 1000 independent samples and   $r=500$ bootstrap resampling. 
The  parametric  setup is identical to the one used in Table \ref{tab:improve} but with fixed sample size equal to $n=100$.
Overall, excluded the independence case, 
%
the simultaneous method \eqref{eq:confidentIntervals} outperforms the pointwise method, since the coverage probabilities are always larger.

%
%
%
\begin{table}[t!]
\begin{center}
{\footnotesize \begin{tabular}{llccccc}
\toprule
Estimator & Confident bands' type & \multicolumn{5}{c}{Parameter $\alpha'$}\\
& & $0.3$ & $0.5$ & $0.7$ & $0.9$ &$1$\\
\midrule
BP-P & Pointwise     & 41.53 & 35.92 & 50.66 & 72.23 & 83.05 \\
& Simultaneous       & 73.34 & 69.13 & 68.79 & 75.11 & 84.82 \\
\midrule
BP-CFG & Pointwise  & 26.89 & 42.65 & 42.60 & 57.68 & 57.30 \\
& Simultaneous        & 62.24 & 61.92 & 60.67 & 66.54 & 57.32\\
\midrule
BP-HT & Pointwise   & 29.33 & 28.92 & 49.38 & 65.20 & 10.42 \\
& Simultaneous       & 51.26 & 54.22 & 60.91 & 81.33 & 10.68 \\
\midrule
BP-MD & Pointwise  & 54.63 & 70.15 & 66.13 & 73.43 & 94.63 \\
& Simultaneous       & 76.40 & 80.48 & 80.26 & 81.36 & 94.65 \\ 
\bottomrule
\end{tabular}}
\caption{$95\%$ coverage  probabilities of the BP method with four non-parametric estimators for  the symmetric logistic model.} 
\label{tab:covprob}
\end{center}
\end{table}
%
 \begin{figure}[t!]
\centering
\includegraphics[width=\textwidth]{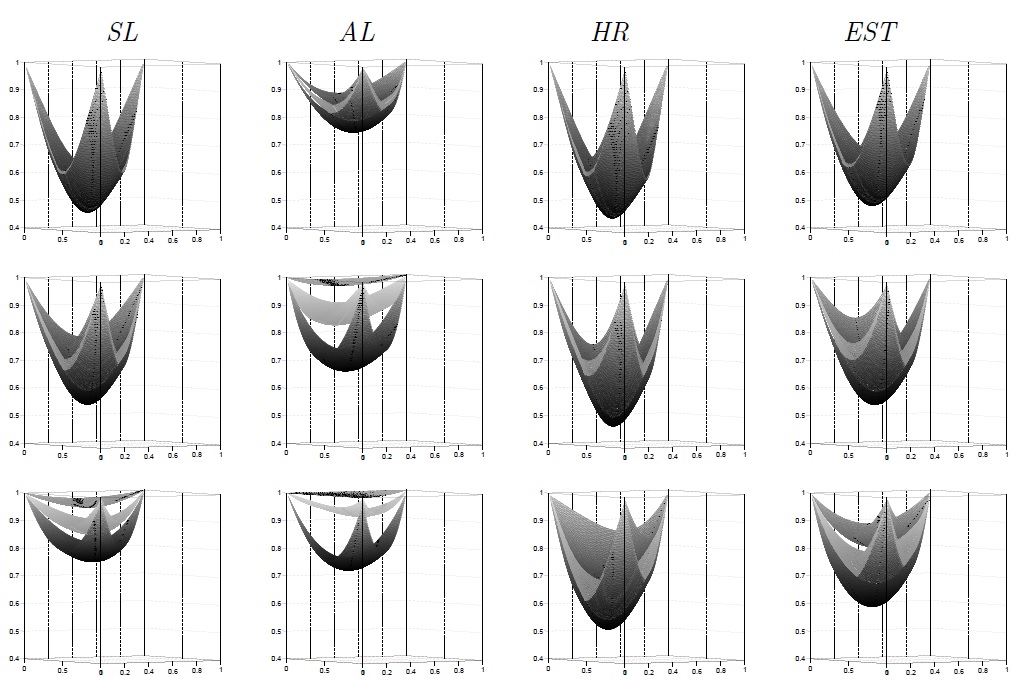}\\
\caption{{ Estimates of  Pickands dependence function for $d=3$ (light grey shade) and bootstrap variability bands (dark grey shade)
for the SL, AL, HR, EST (left-right) models with strong, mild and weak dependence (top-bottom)}}
\label{fig:triv}
\end{figure}
%

To close this small simulation study\footnote{The case $d=2$ has also been considered. The results have been omitted for brevity, since they arrive at the same conclusion. Tables like Table \ref{tab:MISE}, \ref{tab:improve} and 
 \ref{tab:covprob} are available upon request for the HR and EST families and brings the same overall  message.}, 
we extend the class of parametric families to 
the asymmetric logistic (AL, \citeNP{tawn90}) 
with  $\theta=0.6$, $\phi=0.3$, $\psi=0$, 
 the  H\"{u}sler--Reiss model (HR, \citeNP{husler+r89})
  with three cases  ($\gamma_1=0.8,\gamma_2=0.3,\gamma_3=0.7$), ($\gamma_1=0.49, \gamma_2=0.51, \gamma_3=0.03$),
($\gamma_1=0.24,\gamma_2=0.23,\gamma_3=0.11$)   and the  extremal skew-$t$ (EST, \citeNP{padoan11}) with  three setups ($\alpha^{*}=7,-10,1$, $\nu=3$, $\omega=0.9$), 
($\alpha^{*}=-2,9,-15$, $\nu=2$, $\omega=0.9$), ($\alpha^{*}=-0.5,-0.5,-0.5$, $\nu=3$, $\omega=0.9$).
Figure \ref{fig:triv} shows that, for all these cases, the lower and upper limits of the variability bands are always convex functions and
they always contain the true Pickands dependence function. 
The variability bands of weaker  dependence structures are
typically wider than those of stronger dependence structures. 
The same is true for asymmetric versus symmetric dependence structures.

\section{Weekly maxima of hourly rainfall in France}\label{sec:data}

Coming back to Figure \ref{fig:map} introduced in Section \ref{sec:intro}, 
our goal here is to measure the dependence within each cluster of size $d=7$.
The clusters were obtained by running the algorithm proposed by \shortciteN{bernard13} on  
  weekly maxima of hourly rainfall  recorded in the Fall season from 1993 to 2011, i.e., $n=228$ for each station.
In the first place, the aim of clustering was to describe the dependence of locations, with homogeneous climatology characteristics within a cluster and heterogeneous characteristics between clusters.
Climatologically, extreme precipitation  that affects the Mediterranean coast in the fall is caused by the interaction of southern  
and mountains winds coming from the Pyr\'{e}n\'{e}es, C\'{e}vennes and Alps regions. 
In the north of France, heavy rainfall is often produced by
mid-latitude perturbations in Brittany or  in the north of France and Paris. 
It can be checked that extremes within clusters are indeed strongly dependent.

For each cluster, we compute our Bernstein projection estimator based on  the  madogram
and fixed the polynomial's order $k$ equal to 7.
To summarize this seven-dimensional dependence structure, we take advantage of the \emph{extremal coefficient} \cite{smith90} defined by 
$$
\theta=d\,A(1/d,\ldots,1/d).
$$ 
It  satisfies the condition
$1\leq\theta\leq d$, where the lower and upper bounds represent the cases of complete dependence and independence
among the extremes,  respectively. In each cluster, the extremal coefficient is estimated using 
the equation $\widehat{\theta}~=~7\,\widetilde{A}^{\MD}_{n,k}(1/7,\ldots,1/7)$, so that  $\widehat{\theta}$ always belongs to the interval $[1,7]$.
The range of the estimated coefficients is between $3.5$, indicating strong dependence, and $4.6$, indicating medium dependence. 

%
\label{fig:biv-extrindBiv}
\begin{figure}[t]
	\centering
	\includegraphics[width=0.45\textwidth, page=1]{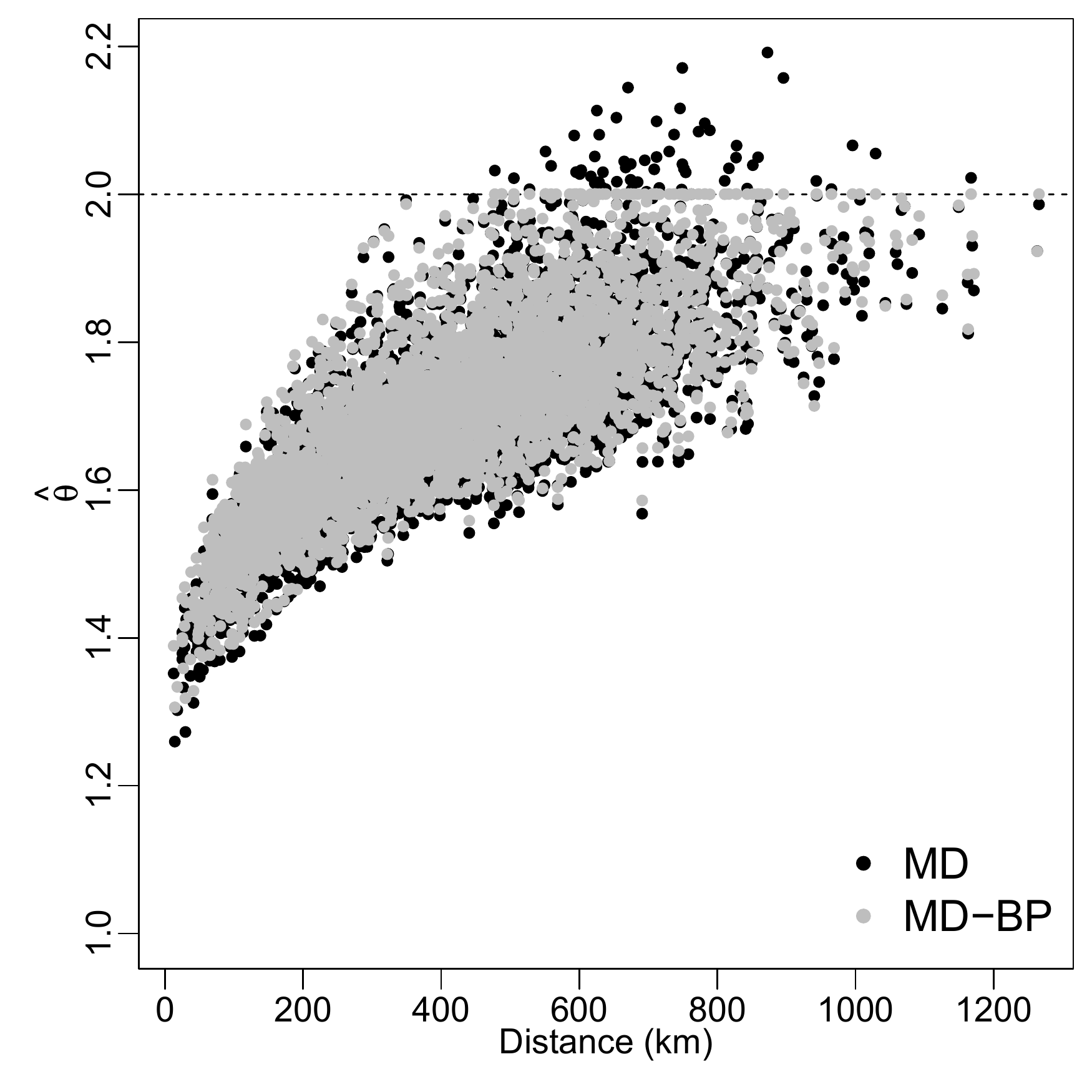} \hspace{.2cm}
	\includegraphics[width=0.45\textwidth, page=2]{Fig5-1}
	\caption{French precipitation data. Left: pairwise extremal coefficients as a function of distance between weather stations. Right: estimates of Pickands dependence functions for four pairs of stations at decreasing distances (black: raw madogram estimator; gray: Bernstein projection madogram estimator).}
\end{figure}
%

As climatologically expected, we can detect in~Figure \ref{fig:map} a  latitudinal gradient in the estimated extremal coefficients. 
They  are smaller in the northern regions and higher in the south. This can be explained by westerly fronts  above $46^\circ$ latitude that affect large regions, whereas  extreme precipitation in the south is more likely to be driven by localised convective storms with  weak spatial dependence structures.
Finally, in the center of the country, away from the coasts, there is the highest degree of dependence among extremes, as they are the result of the meeting between different densities of air masses.

For all possible pairs of locations we have estimated the bivariate Pickands dependence function using the madogram estimator and its Bernstein projection.
The left-hand panel of Figure \ref{fig:biv-extrindBiv} shows the pairwise extremal coefficients versus the Euclidean distance between sites, computed through the estimated Pickands dependence functions. We have
$\widehat{\theta}\leq 1.5$ for the locations that are less than 200 km far apart, meaning that the extremes are strongly or at least mildly dependent, while for sites more than 200 km far apart, we have $\widehat{\theta}>1.5$, meaning that the extremes at most mildly dependent up to independent.
The graph also shows the benefits of the projection method: after projection, the extremal coefficients fall within the admissible range $[1,2]$, whereas they can be larger than $2$ without the projection method.

The right-hand plot of Figure~\ref{fig:biv-extrindBiv} shows four examples of estimated Pickands dependence functions obtained with pairs of sites whose distances are 979.8, 505.9, 390.1 and 158.1 km, respectively (top-left to bottom-right panels). 
The madogram estimator provides estimates (black lines) that are not convex functions and hence are not Pickands dependence functions themselves.
Contrarily, the estimates (gray lines) obtained with the projection estimator are valid Pickands dependence functions.

\section{Computational Details}\label{sec:comp}

Simulations and data analysis were performed using the \textsf{R} package \texttt{ExtremalDep} (\url{https://r-forge.r-project.org/R/?group\_id=1998})

\section*{Acknowledgements}
We thank three anonymous Referees for their valuable suggestions that have improved the presentation of this paper.
We also thank Catia Scricciolo, Sabrina Vettori and Sonia Petrone for their valuable support.
The work of Philippe Naveau has been supported by ANR-DADA, LEFE-INSU-Multirisk, AMERISKA, A2C2, and Extremoscope projects
and was completed  during his visit at the IMAGE-NCAR group in Boulder, CO, USA. 
The work of Johan Segers was funded by contract ``Projet d'Act\-ions de Re\-cher\-che Concert\'ees'' No.\ 12/17-045 of the ``Communaut\'e fran\c{c}aise de Belgique'' and by IAP research network Grant P7/06 of the Belgian government (Belgian Science Policy).
The authors acknowledge Meteo France for the precipitation time series. 
The authors would also very much like to credit the contributors to the \textsf{R} project and \LaTeX. 

\appendix

\section{Proofs}\label{sec:appA}
For $\bw \in \simp$, define the function $\nu_{\bw}: [0,1]^d\rightarrow [0,1]$ by
\begin{equation}\label{eq:mado_function}
\nu_{\bw}(\bu)=\bigvee_{i=1}^d(u_i^{1/w_i})-\frac{1}{d}\sum_{i=1}^{d} u_i^{1/w_i}, \quad \bu\in[0,1]^d,
\end{equation}
where, by convention, $u^{1/w} = 0$ whenever $w = 0$ and $u \in [0, 1]$.

\begin{lem}
\label{lem:int}
For any cumulative distribution function $H$ on $[0, 1]^d$ and for any $\bw \in \simp$, we have
\[
  \int_{[0, 1]^d} \nu_{\bw}(\bu) \, \diff H(\bu)
  =
  \frac{1}{d} \sum_{i=1}^d \int_0^1 H(1, \ldots, 1, x^{w_i}, 1, \ldots, 1) \, \diff x
  -
  \int_0^1 H(x^{w_1}, \ldots, x^{w_d}) \, \diff x.
\]
\end{lem}

\begin{proof}
Fix $\bw\in\simp$. For every $\bu\in [0,1]^d$ we have
\begin{equation*}
\begin{split}
\bigvee_{i=1}^du_i^{1/w_i}
&=1-\int_0^1\indic(\forall i=1,\ldots,d:u_i^{1/w_i}\leq x)dx\\
&=1-\int_0^1\indic(\forall i=1,\ldots,d:u_i\leq x^{w_i})dx
\end{split}
\end{equation*}
and
$$
\frac{1}{d}\sum_{i=1}^d u_i^{1/w_i}=1-\frac{1}{d}\sum_{i=1}^d\int_0^1\indic(u_i\leq x^{w_i})dx.
$$
Subtracting both expressions and integrating over $H$ yields
\begin{equation*}
\begin{split}
  \int_{[0, 1]^d} \nu_{\bw}(\bu) \, \diff H(\bu)
  &= 
  \frac{1}{d} \sum_{i=1}^d \int_{[0, 1]^d} \int_0^1 \indic(u_i \le x^{w_i}) \, \diff x \, \diff H(\bu) \\
  &\qquad \mbox{}
  - \int_{[0, 1]^d} \int_0^1 \indic(\forall i = 1, \ldots, d: u_i \le x^{w_i}) \, \diff x \, \diff H(\bu).
\end{split}
\end{equation*}
Applying Fubini's theorem to both double integrals yields the stated formula.
\end{proof}

\begin{proof}[Proof of Proposition~\ref{prop:multimado}]
The marginal distribution functions being continuous, the copula $C$ is the joint distribution function of the random vector $(F_1(X_1), \ldots, F_d(X_d))$. For $\bw \in \simp$, the multivariate $\bw$-madogram can thus be written as
\[
  \nu(\bw) = \int_{[0, 1]^d} \nu_{\bw}(\bu) \, \diff C(\bu).
\]
Next, apply Lemma~\ref{lem:int}. 
Since $C$ is an extreme-value copula with Pickands dependence function $A$, we find, after some elementary calculations using \eqref{eq:ev_copula} and \eqref{eq:ellA},
\[
  C(x^{w_1}, \ldots, x^{w_d}) = x^{A(\bw)}
\]
for all $x \in (0, 1)$. We obtain
\begin{eqnarray}
\label{eq:C2nu}
  \nu(\bw) 
  &=& \frac{1}{d} \sum_{i=1}^d \int_0^1 C(1, \ldots, 1, x^{w_i}, 1, \ldots, 1) \, \diff x 
  -
  \int_0^1 C(x^{w_1}, \ldots, x^{w_d}) \, \diff x
  \\
\nonumber
  &=& \frac{1}{d} \sum_{i=1}^d \int_0^1 x^{w_i} \, \diff x - \int_0^1 x^{A(\bw)} \, \diff x,
\end{eqnarray}
yielding the first formula stated in the proposition. 
Solve for $A(\bw)$ to obtain \eqref{eq:Amd}. Since $\nu(\bw) + c(\bw) = A(\bw) / (1 + A(\bw))$, necessarily $\nu(\bw) + c(\bw) < 1$, so that the right-hand side of \eqref{eq:Amd} is well-defined.
\end{proof}

\begin{proof}[Proof of Theorem~\ref{prop:prop_multimado}]
The proof proceeds by expressing the statistics and empirical $\bw$-madogram $\nu_n( \bw )$ and $\widehat{\nu}_n( \bw )$ in terms of the empirical distribution and empirical copula and exploiting known results thereon.
For $i = 1, \ldots, d$ and $j = 1, \ldots, n$, let
\begin{equation*}
\begin{split}
\bU_j&= (U_{j,1},\ldots,U_{j,d}),\quad U_{j,i}=F_{i}(X_{j,i}),\\
\widehat{\bU}_j
  &= (\widehat{U}_{j,1},\ldots,\widehat{U}_{j,d}), \quad
  \widehat{U}_{j,i}= F_{n,i}(X_{j,i})
  = \frac{1}{n}\sum_{m=1}^n\indic(X_{m,i}\leq X_{j,i}).
\end{split}
\end{equation*}
Recall $\nu_{\bw}$ in \eqref{eq:mado_function}. The statistics and empirical $\bw$-madogram are equal to
$$
\nu_n( \bw ) = \frac{1}{n} \sum_{m=1}^n \nu_{\bw}(\bU_m)
  = \int_{[0, 1]^d} \nu_{\bw}( \bu ) \, \diff C_n(\bu),\quad
\widehat{\nu}_n( \bw ) =
  \int_{[0, 1]^d} \nu_{\bw}( \bu ) \, \diff \hat{C}_n(\bu),  
$$
respectively, 
where $C_n$ and $\widehat{C}_n$ are the empirical distribution and copula:
$$
C_n(\bu) = \frac{1}{n} \sum_{m=1}^n \indic(\bU_m \le \bu ), \quad
\widehat{C}_n(\bu) = \frac{1}{n} \sum_{m=1}^n \indic( \widehat{\bU}_m \le \bu ),
\qquad \bu \in [0, 1]^d,
$$
(component-wise inequalities).
By Lemma~\ref{lem:int} we obtain
\begin{equation}\label{eq:Cn2nun}
\nu_n(\bw)
=
\frac{1}{d} \sum_{i=1}^d \int_0^1 C_n(1, \ldots, 1, x^{w_i}, 1, \ldots, 1) \, \diff x
-
\int_0^1 C_n(x^{w_1}, \ldots, x^{w_d}) \, \diff x,
\end{equation}
and a similar expression is attained for $\widehat{\nu}_n(\bw)$ but with $C_n$ replaced by $\widehat{C}_n$.
Comparing the latter equation with \eqref{eq:C2nu} yields
$$
\norm{\nu_n - \nu }_\infty \le 2 \norm{C_n - C }_\infty.
$$
Standard empirical process arguments yield uniform strong consistency of the empirical copula (\citeNP{deheuvels91}). We come to a similar inequality for $\widehat{\nu}_n$. 
Uniform strong consistency of $A_n$ and $\widehat{A}_n$ follows.

Next, consider the empirical processes
$$
\CB_n=\sqrt{n}(C_n-C), \qquad \copula_n=\sqrt{n}(\widehat{C}_n-C).
$$
Combining Equations~\eqref{eq:C2nu} and \eqref{eq:Cn2nun} we obtain
$$
\sqrt{n} \bigl(\nu_n(\bw) - \nu(\bw) \bigr)
=\frac{1}{d}\sum_{i=1}^d \int_0^1\CB_n(1,\ldots,1,x^{w_i},1,\ldots,1) \diff x
-\int_0^1\CB_n(x^{w_1},\ldots,x^{w_d})\diff x
$$
and clearly a similar expression is obtained for $\sqrt{n} \bigl( \widehat{\nu}_n(\bw) - \nu(\bw) \bigr)$
but replacing $\CB_n$ with $\copula_n$.
Now, two related results: in the space $\ell^{\infty}([0,1]^d)$ equipped with the supremum norm, 
$\CB_n\conW\CB$, as $n\rightarrow\infty$, where $\CB$ is a C-Brownian bridge, and if Condition \ref{cond:smooth} holds, then $\copula_n\conW\copula$, as $n\rightarrow\infty$, 
where $\copula$ is the Gaussian process defined in \eqref{eq:cop_proc}.
The map
\begin{equation*}
\phi : \ell^{\infty}([0,1]^d) \to \ell^{\infty}(\simp) : f \mapsto \phi(f)
\end{equation*}
defined by
$$
(\phi(f))(\bw)
= \frac{1}{d}\sum_{i=1}^d \int_0^1 f(1,\ldots,1,x^{w_i},1,\ldots,1) \, \diff x
-
\int_0^1 f(x^{w_1},\ldots,x^{w_d}) \, \diff x
$$
is linear and bounded, and therefore continuous. The continuous mapping theorem then implies
$$
\sqrt{n}(\nu_n-\nu)=\phi(\CB_n)\conW\phi(\CB),
\quad\sqrt{n}(\widehat{\nu}_n-\nu)=\phi(\copula_n)\conW\phi(\copula),\quad n\rightarrow\infty,
$$
in $\ell^{\infty}(\simp)$. The Gaussian process $\copula$ satisfies
$$
\prob\{\forall\, i=1,\ldots,d :\forall\, u \in[0,1]: \copula(1,\ldots,1,u,1,\dots,1)=0\}=1.
$$
This property follows from the continuity of its sample paths and by the form of the
covariance function \eqref{eq:covariance}. We find, for $\bw \in \simp$,
$$
(\phi(\copula))(\bw) = - \int_0^1 \copula(x^{w_1}, \ldots, x^{w_d}) \, \diff x.
$$
Finally, apply the functional delta method (\citeNP{vaart98}, Ch.\ 20) to arrive at the conclusion.
\end{proof}

\begin{proof}[Proof of Proposition~\ref{prop:conv_bapp}]
We have $|B_A(\bw;k)-A(\bw)|\leq\expect |A(\bY_k/k)-A(\bw)|$, where
$\bY_k=(Y_{k,i};i=1,\ldots,d)$ is a multinomial random vector with $k$ trials, $d$ possible outcomes, and success probabilities $w_1,\ldots,w_d$.
Any function $A \in \spA$ is Lipschitz-1, so that
$|B_A(\bw;k)-A(\bw)|\leq\sum_{i=1}^d\expect|Y_{k,i}/k-w_i|$. By the Cauchy--Schwarz inequality and the fact that the random variables $Y_{i,k}$ are binomially distributed, it follows that
$|B_A(\bw;k)-A(\bw)|\leq\sum_{i=1}^d(\expect(Y_{k,i}/k-w_i)^2)^{1/2} \le d/(2\sqrt{k})$.
%
\end{proof}
%
%
%
%
\begin{proof}[Proof of Proposition~\ref{prop:lower}]
On the one hand we have that if $B_A(\bw;k) \geq \max(w_1, \ldots, w_d)$, then $D_{\bv_i-\bv_j} B_A(\bv_j;k) \geq -1$. Indeed, $\max(w_1, \ldots, w_d)$ is the intersection of the planes $z_{0} = 1-w_1-w_2-\cdots-w_{d-1}$, $z_1=w_1$, $\ldots$, $z_{d-1}=w_{d-1}$, then by the assumption
$$
B_A(\bv_j;k) \geq z_j, \qquad j=0,1,\ldots,d-1.
$$
The directional derivatives of $B_A$ calculate for $\bv_j$, $j=0,1,\ldots,d-1,$ are equal to
\begin{equation} \label{eq:directderiv}
D_{\bv_i-\bv_j} B(\bv_j;k)=
\begin{cases}
D_{\bv_i-\bv_0} B(\bv_0;k) & \mbox{ if } i \neq 0 = j\\
- D_{\bv_j-\bv_0} B(\bv_j;k) & \mbox{ if } i = 0 \neq j\\
D_{\bv_i-\bv_0} B(\bv_j;k) - D_{\bv_j-\bv_0} B(\bv_j;k) & \mbox{ if } i\neq 0 \neq j, i\neq j\\
\end{cases}
\end{equation}
which are bounded from below by $-1$.
Then, considering the directional derivatives on both sides of the above inequality we obtain
$$
D_{\bv_i-\bv_j} B_A(\bv_j;k) \geq -1,\quad \forall\;i,j = 0,1,\ldots,d-1, \, i\neq j,
$$
and hence the result.

On the other hand if
$D_{\bv_i-\bv_j}B_A(\bv_j;k)\geq -1$, $j=0,\ldots,d-1$ then $B_A(\bw;k)\geq\max(w_1,\ldots,w_d)$. 
Since $B_A$ lies above the tangent plane
\begin{equation} \label{eq: tangentplane}
B_A(\bw;k) \geq B_A(\bw';k) + (\bw'-\bw)^\top\nabla B_A(\bw';k), \qquad \forall \, \bw,\bw' \in \simp.
\end{equation}
by the convexity assumption, then evaluating this inequality for $\bw'=\bv_j$ for $j \in \{ 0,1,\ldots,d-1\}$ we obtain the desired result $B_A(\bw;k) \geq w_j$ for all $\bw \in \simp$. Indeed, 
considering \eqref{eq: tangentplane} at $\bw'=\bv_0$ we find, for $\bw \in \simp$,
$$
B_A(\bw;k) \geq 1 + \bw^\top \nabla B_A(\bv_0;k) = 1 + \sum_{i=1}^{d-1} w_i \,  D_{\bv_i-\bv_0} B(\bv_0;k) \geq 1 + \sum_{i=1}^{d-1} w_i \, (-1) = w_d
$$
where $w_d = 1-w_1-\cdots-w_{d-1}$, as required. Furthermore, considering \eqref{eq: tangentplane} at $\bw'=\bv_j$ for $j\in \{1,\ldots,d-1\}$ we find for $\bw \in \simp$,
\begin{eqnarray} \nonumber
B_A(\bw;k) &\geq& 1 + (\bw-\bv_j)^\top \nabla B_A(\bv_j;k) \\ \nonumber
&=& 1 + (w_j-1)D_{\bv_j-\bv_0} B(\bv_j;k) + \sum_{\substack{i=1\\i\neq j}}^{d-1} w_i \,  D_{\bv_i-\bv_0} B(\bv_j;k) \\ \nonumber
&\geq& 1 + (w_j-1)D_{\bv_j-\bv_0} B(\bv_j;k) + \sum_{\substack{i=1\\i\neq j}}^{d-1} w_i \, \bigg( D_{\bv_j-\bv_0} B(\bv_j;k) -1 \bigg) \\ \nonumber
&=& 1 + (w_j-1)D_{\bv_j-\bv_0} B(\bv_j;k) + (1-w_j-w_d) \, \bigg( D_{\bv_j-\bv_0} B(\bv_j;k) -1 \bigg) \\ \nonumber
&=& w_j + w_d \, \bigg(1-D_{\bv_j-\bv_0} B(\bv_j;k) \bigg)\geq w_j \nonumber
\end{eqnarray}
given that $D_{\bv_i-\bv_0} B(\bv_j;k) \geq D_{\bv_j-\bv_0} B(\bv_j;k) -1$ and $1-D_{\bv_j-\bv_0} B(\bv_j;k) \geq 0$. 
\end{proof}
%
%
%
%
\begin{proof}[Proof of Proposition~\ref{prop:nested}]
Firstly, the polynomials in $\spA_k$ are nested (e.g., \citeANP{wang+g12}, \citeyearNP{wang+g12}; 
\citeANP{farin86}, \citeyearNP{farin86}).
By the degree-raising property we have 
$$
B(\bw;k)=\sum_{\balpha \in \Gamma_k} \beta_{\balpha} b_{\balpha}(\bw;k)=\sum_{\balpha \in \Gamma_{k+1}} \tilde{\beta}_{\balpha} b_{\balpha}(\bw;k+1)=\tilde{B}(\bw;k+1)
$$
where 
\begin{equation}
\label{eq:beta_tilde}
\tilde{\beta}_{\balpha} = \sum_{h=1}^d \frac{\alpha_{h}}{k+1}\beta_{\balpha - \bv_{h-1}}.
\end{equation}
We need to show that the coefficients $\tilde{\beta}_{\balpha}$ satisfy the constraints R1)-R2)-R3).
For the case R1) we need to check that
$$
\Delta_{i,0}^2 \tilde{\beta}_{\balpha} - \sum_{\substack{j\neq i}} |\Delta_{i,0} \Delta_{j,0} \tilde{\beta}_{\balpha}|\geq 0,\quad \forall \balpha \in \Gamma_{(k+1)-2},\;i=1,\ldots,d-1.
$$
This can be rewritten as
$$
\Delta_{i,0}^2 \tilde{\beta}_{\balpha} - \sum_{\substack{j\neq i}} (-1)^{I_{s,t}}\Delta_{i,0} \Delta_{j,0} \tilde{\beta}_{\balpha}\geq 0,
$$
where $I_{s,t}$ is the set of all the possible combinations with repetition of the set $\{1,2\}$ in sequences of $d-2$ terms,
$s=1,\ldots,d-2$ and $t=1,\ldots,2^{d-2}$.
Using the relation in \eqref{eq:beta_tilde} we have
$$
\tilde{\beta}_{\balpha} = \sum_{h=1}^d \frac{\alpha_{h}}{k+1} \beta_{\balpha-\bv_{h-1}} = \sum_{h=1}^{d-1} \frac{\alpha_{h}}{k+1} \beta_{\balpha-\be_h} + \frac{\alpha_{d}}{k+1} \beta_{\balpha}.
$$
Then, we obtain
\begin{eqnarray} \nonumber
\Delta_{i,0}^2 \tilde{\beta}_{\balpha} - \sum_{\substack{j\neq i}} (-1)^{I_{s,t}}\Delta_{i,0} \Delta_{j,0} \tilde{\beta}_{\balpha} \nonumber
&=&
\Delta_{i,0}^2 \left\{ \sum_{h=1}^{d-1} \frac{\alpha_{h}}{k+1} \beta_{\balpha-\be_h} + \frac{\alpha_{d}}{k+1} \beta_{\balpha} \right\}\\ \nonumber
&-&  \sum_{\substack{j\neq i}} (-1)^{I_{s,t}} \Delta_{i,0} \Delta_{j,0}  \left\{  \sum_{h=1}^{d-1} \frac{\alpha_{h}}{k+1} \beta_{\balpha-\be_h} + \frac{\alpha_{d}}{k+1} \beta_{\balpha} \right\}\\ \nonumber
&=&
\sum_{h=1}^{d-1} \frac{\alpha_{h}}{k+1} \left\{ \Delta_{i,0}^2 \beta_{\balpha-\be_h} - \sum_{\substack{j\neq i}} (-1)^{I_{s,t}} \Delta_{i,0} \Delta_{j,0} \beta_{\balpha-\be_h} \right\} \\ \nonumber
	&+& \sum_{h=1}^{d-1} \frac{\alpha_{d} -1}{k+1} \left\{ \Delta_{i,0}^2 \beta_{\balpha} - \sum_{\substack{j\neq i}} (-1)^{I_{s,t}} \Delta_{i,0} \Delta_{j,0} \beta_{\balpha} \right\} \geq 0,\\ \nonumber
\end{eqnarray}
and hence the result.
For the case R2), using \eqref{eq:beta_tilde}, it is immediate to verify for the set
$
\{ \tilde{\beta}_{\balpha}, \balpha \in \Gamma_{k+1}
: \balpha = \bzero \text{ or }  \balpha = (k+1) \, \be_i, \, \forall i=1,\dots, d-1 \}
$
that $\tilde{\beta}_{\balpha}=\beta_{\balpha}=1$. 
Finally, for the case R3) we need to check that
$
1 - 1/(k+1) < \tilde{\beta}_{\balpha},
$
where
$
\{ \tilde{\beta}_{\balpha}, \balpha \in \Gamma_{k+1}
: \balpha = \be_i \text{ or }  \balpha = k \be_i  
\text{ or } \balpha_l = k \be_i + \be_j, \; \forall j\neq i=1,\dots, d-1 \}.
$
By definition we have
$$
\tilde{\beta}_{\be_i} = \frac{k}{k+1}\beta_{\be_i} + \frac{1}{k+1}, \;
\tilde{\beta}_{k\be_i} = \frac{k}{k+1}\beta_{(k-1)\be_i} + \frac{1}{k+1},\;
\tilde{\beta}_{k\be_i+\be_j} = \frac{k}{k+1}\beta_{(k-1)\be_i+\be_j} + \frac{1}{k+1}.
$$
Substituting $\tilde{\beta}_{\balpha}$, with $\balpha = \be_i, \balpha = k \be_i, \balpha = k \be_i+\be_j$, in the previous inequality we obtain
\begin{eqnarray} \nonumber
\tilde{\beta}_{\balpha} &\geq& 1-\frac{1}{k+1} \\
\nonumber
\frac{1}{k+1} \left\{ 1+k\beta_{\balpha-\bv_{i-1}} \right\} &\geq & \frac{1}{k+1} \left\{ 1+k\left(1-\frac{1}{k}\right) \right\} \\ \nonumber
\beta_{\balpha-\bv_{i-1}}&\geq & 1 - \frac{1}{k}  
\end{eqnarray}
for $i=1,\ldots,d-1$, and hence the result. 
Thus the first statement is proven.\par
Secondly, let $A$ be a Pickands dependence function and consider the Bernstein polynomial
\[
  A_k( \bw ) = \sum_{\balpha \in \Gamma_k} A( \balpha / k ) \, b_{\balpha}( \bw ; k ),
\]
that is, $A_k = B_A( \,\cdot\,; k )$ as in \eqref{eq:polyrap}. Constraint R1) holds by assumption~\eqref{eq:wddA}. Since $\max(w_1, \ldots, w_d) \le A( \bw ) \le 1$ for all $\bw \in \simp$, the constraints in R2) and R3) are satisfied too. Finally, we have uniform convergence $A_k \to A$ by Proposition~\ref{prop:conv_bapp}.
\end{proof}

\begin{proof}[Proof of Proposition~\ref{prop:prop_bernproj}]
Consider the projection of the madogram estimator on the full space $\spA$ (rather than on the subspace $\spA_k$): 
\[
  \widetilde{A}_n^{\MD}
  = 
  \argmin_{B \in\spA} \norm{ \widehat{A}_n^{\MD}-B}_2.
\]
From Theorem \ref{prop:prop_multimado} it follows that $\sqrt{n}(\widehat{A}_n^{\MD}-A)\conW Z$ in $L^2(\simp)$ as $n\rightarrow\infty$ where $Z$ is a Gaussian process. Theorem~1 in \shortciteN{fil+g+s08} then implies that
\[
  \sqrt{n} ( \widetilde{A}_{n}^{\MD} - A )
  \conW
  \argmin_{Z'\in T_\spA(A)}\|Z'-Z\|_2,\quad n\rightarrow \infty.
\]
It remains to show that we can replace $\widetilde{A}_n^{\MD}$ by $\widetilde{A}_{n,k_n}^{\MD}$. It suffices to show that
\[
  \norm{ \widetilde{A}_{n,k_n}^{\MD} - \widetilde{A}_{n}^{\MD} }_2
  = o_p(n^{-1/2}), \qquad n \to \infty.
\]
By the first inequality in Lemma~1 in \shortciteN{fil+g+s08} with, in their notation, $\mathcal{F} = \mathcal{A}$ and $\mathcal{G} = \mathcal{A}_k$, we find that
\[
  \norm{ \widetilde{A}_{n,k_n}^{\MD} - \widetilde{A}_{n}^{\MD} }_2
  \le
  [ \delta_{k_n} ( 2 \norm{ \widehat{A}_{n}^{\MD} - \widetilde{A}_{n}^{\MD} }_2 + \delta_{k_n} ) ]^{1/2},
\]
where $\delta_{k_n}$ is bounded by the $L_2$ Hausdorff distance between $\spA$ and $\spA_{k_n}$. Proposition~\ref{prop:conv_bapp} yields $\delta_{k_n} = O(k_n^{-1/2})$, which is $o(n^{-1/2})$ by the assumption on $k_n$. Furthermore, since $A \in \spA$, we find, by definition of the projection estimator,
\[
  \norm{ \widehat{A}_{n}^{\MD} - \widetilde{A}_{n}^{\MD} }_2
  \le
  \norm{ \widehat{A}_{n}^{\MD} - A }_2
  =
  O_p(n^{-1/2}), \qquad n \to \infty.
\]
Combine these relations to complete the proof.
\end{proof}

\bibliographystyle{chicago}
\bibliography{bernpoly}
\end{document}

%% file: characters.tex
\def\thick#1{\hbox{\rlap{$#1$}\kern0.25pt\rlap{$#1$}\kern0.25pt$#1$}}

\def\bY{{\bf Y}}

\def\pone{\text{C1}}
\def\ptwo{\text{C2}}

\def\rone{\text{R1}}
\def\rtwo{\text{R2}}
\def\rthree{\text{R3}}
\def\simp{\mathcal{S}_{\ds}}
\def\ds{{d-1}}

\def\balpha{{\boldsymbol\alpha}}
\def\bbeta{{\boldsymbol\beta}}
\def\bgamma{{\thick\gamma}}

\def\bzero{{\bf 0}}

\def\indic{\mathds{1}}

\def\real{{\mathbb R}}

\def\expect{{\mathbb E}}
\def\prob{{\mathbb P}}

\def\nat{{\mathbb N}}
\def\copula{\widehat{\mathbb{D}}}

\def\spA{\mathcal{A}}

\def\spC{\mathcal{C}}

\def\spX{\mathcal{X}}

\def\conAS{\stackrel{\mathrm{a.s.}}{\longrightarrow}}
\def\conD{\Rightarrow}
\def\conW{\leadsto}

\def\CB{\mathbb{D}}

\def\MD{\mathrm{MD}}
\newcommand{\norm}[1]{\lVert{#1}\rVert}
\newcommand{\abs}[1]{\lvert{#1}\rvert}
\newcommand{\Cov}{\operatorname{Cov}}
\newcommand{\diff}{\mathrm{d}}

\newcommand{\bb}{{\boldsymbol{b}}}
\newcommand{\be}{{\boldsymbol{e}}}
\newcommand{\br}{{\boldsymbol{r}}}
\newcommand{\bu}{{\boldsymbol{u}}}
\newcommand{\bv}{{\boldsymbol{v}}}
\newcommand{\bw}{{\boldsymbol{w}}}
\newcommand{\bx}{{\boldsymbol{x}}}

\newcommand{\bz}{{\boldsymbol{z}}}
\newcommand{\bR}{{\boldsymbol{R}}}
\newcommand{\bX}{{\boldsymbol{X}}}
\newcommand{\bU}{{\boldsymbol{U}}}

\newcommand{\argmin}{\operatornamewithlimits{\arg\min}}